\documentclass[journal, 10pt, twocolumn]{IEEEtran}
\usepackage{amsmath,amsfonts,amssymb, amsthm}

\usepackage{algorithm}
\usepackage[noend]{algorithmic}
\usepackage{optidef} %for optimization
\usepackage{array}
\usepackage[dvipsnames]{xcolor}
\usepackage{textcomp}
\usepackage{stfloats}
\usepackage{url}
\usepackage{tabularx, booktabs, multirow, makecell}
\usepackage{threeparttable} 
\usepackage{bbding}
\usepackage{lineno}
\usepackage{cite}
\usepackage{enumitem}
\usepackage{booktabs} 
\usepackage{multirow} 
\usepackage{graphicx} 
\usepackage{siunitx} 
\usepackage[normalem]{ulem}
\newtheorem{proposition}{Proposition}



\usepackage{graphicx}
\ifCLASSOPTIONcompsoc
    \usepackage[caption=false, font=normalsize, labelfont=sf, textfont=sf]{subfig}
\else
\usepackage[caption=false, font=footnotesize]{subfig}
\fi

\makeatletter
\renewcommand\p@subfigure{}% 不自動加主圖編號
\renewcommand{\subref}[1]{(\ref{#1})}
\makeatother

\def\BibTeX{{\rm B\kern-.05em{\sc i\kern-.025em b}\kern-.08em
    T\kern-.1667em\lower.7ex\hbox{E}\kern-.125emX}}
\usepackage{balance}

\begin{document}
\title{Joint Age-of-Latent and Resource Minimization for Wireless Multi-Camera Perception With Temporal Window Selection}

\author{
Chih-Yu Lin\IEEEauthorrefmark{1},
Wanjiun Liao\IEEEauthorrefmark{1},
Sumudu Samarakoon\IEEEauthorrefmark{2},
and Mehdi Bennis\IEEEauthorrefmark{2}
\\
\IEEEauthorrefmark{1}Department of Electrical Engineering,
National Taiwan University, Taipei, Taiwan
\\
\IEEEauthorrefmark{2}University of Oulu, Oulu, Finland
\\
E-mails: \{d09921030,wjliao\}@ntu.edu.tw,
\{sumudu.samarakoon,mehdi.bennis\}@oulu.fi\thanks{This work has been submitted to the IEEE for possible publication. Copyright may be transferred without notice, after which this version may no longer be accessible.}
}
\maketitle

\begin{abstract}
Multi-camera wireless perception requires a base station (BS) to maintain timely and reliable latent beliefs from distributed cameras under limited uplink resources. Conventional Age-of-Information (AoI) measures the age of the latest received update but not task-relevant latent content. We introduce Age-of-Latent (AoL) to quantify the freshness of each camera's latest decoded latent representation. A finite temporal window of integration (TWI) determines the task-commitment time and available uplink slots, creating a tradeoff among update opportunities, prediction duration, commit-time AoL, prediction reliability, and accumulated resource cost.
Within this horizon, redundant or overlapping views enable correlated prediction, reducing reliance on the highest-cost communication and encoding configuration to maintain BS-side latent beliefs.
We formulate a joint AoL-resource minimization problem coupling task-level TWI selection with slot-level encoder selection, scheduling, and NOMA power allocation under prediction-reliability constraints.
We propose correlation-aware latent prediction for AoL minimization (CoLA), which uses Lyapunov optimization for task-level TWI selection based on AoL-resource cost and prediction uncertainty, and proximal policy optimization for slot-level resource control.
Results on a warehouse multi-camera RF dataset show that CoLA adapts the TWI to camera-update availability and achieves the most favorable AoL-resource tradeoff among the benchmarks while maintaining prediction reliability, particularly under prolonged and severe camera outages.
\end{abstract}

\begin{IEEEkeywords}
Age-of-Latent (AoL), temporal window of integration (TWI), JEPA, semantic communication, NOMA.
\end{IEEEkeywords}

\section{Introduction}
Sixth-generation (6G) networks are expected to support AI-native edge intelligence, in which wireless links enable inference, control, and real-time decision-making in dynamic physical environments. Reliable bit delivery remains necessary in such systems. However, conventional communication metrics such as throughput, latency, and spectral efficiency do not directly measure whether the received information is useful for the downstream task. Perception-driven applications further require the receiver to maintain a timely and reliable task-relevant state, rather than merely reconstructing transmitted packets. Task-oriented and semantic communication therefore provide a promising paradigm for optimizing wireless transmission according to the task value of information, rather than solely pursuing bit-level reliability \cite{ref:goal_oriented_6g_goals,ref:freshness_to_semantics,ref:semcom_standardization_2025}.

The rise of physical AI is reshaping edge intelligence from content consumption toward continuous perception and real-time interaction with physical environments. In visual perception-based control applications such as warehouse robotics, industrial automation, and traffic management, distributed cameras provide visual observations for real-time scene understanding and decision-making. Recent industry reports indicate that AI-enabled services are expected to alter mobile traffic patterns by increasing uplink demand and latency sensitivity \cite{ref:nokia_physical_ai_ran_2026}. Transmitting raw observations over wireless links can therefore be resource-intensive, while reconstructing all source data may be unnecessary for task execution. A more efficient approach is to communicate compact task-oriented features that preserve information needed for collaborative perception \cite{ref:semantic_collaborative_perception_2025}. Recent works further show that predictive representations can model future system states in an abstract latent space \cite{ref:latent_multimodal_dynamics,ref:jepa_msac_2026,ref:wireless_jepa_2026}. In particular, joint embedding predictive architectures (JEPAs) learn such representations without reconstructing pixel-level observations \cite{ref:ijepa}, thereby providing a semantic-level interface for downstream control under limited communication resources \cite{ref:latent_multimodal_dynamics,ref:jepa_msac_2026}. JEPA-based representation learning has also been explored for spatio-temporal latent prediction from multi-antenna wireless signals \cite{ref:wireless_jepa_2026} and multiagent V2X perception \cite{ref:v2x_jepa_2026}. Beyond the spatio-temporal structure of wireless signals, multi-camera perception introduces cross-view correlation among cameras observing the same physical scene from different viewpoints. 
Such correlation enables the receiver to predict a camera's latent representation using information from other cameras.

Age-of-Information (AoI) provides a fundamental measure of information freshness by quantifying the time elapsed since the generation of the latest successfully received update \cite{ref:aoi_survey_yates}. However, packet-level freshness alone is insufficient when a base station (BS) maintains task-relevant latent beliefs through prediction. Although prediction can maintain a useful BS-side latent belief without a fresh update, it does not yield a newly decoded latent representation. Therefore, freshness should be characterized at the latent level. We refer to this metric as Age-of-Latent (AoL), which measures the freshness of the latest decoded latent representation.

Beyond latent freshness, semantic decision-making also requires the BS-side latent beliefs to remain temporally valid at the decision time. The temporal window of integration (TWI) defines a finite horizon for integrating distributed sensory updates and determines when the maintained beliefs are committed to the downstream task~\cite{ref:twi_physical_ai}. The TWI length determines the number of available update opportunities and the duration over which prediction may be required. These coupled effects motivate adaptive TWI selection based on latent freshness and prediction reliability.

Limited radio resources pose an additional challenge to maintaining low AoL within a finite TWI. Multiple cameras may need to upload latent representations before task commitment, while orthogonal access provides only a limited number of update opportunities. Non-orthogonal multiple access (NOMA) can increase update opportunities by allowing concurrent latent transmissions over shared time-frequency resources \cite{ref:noma_survey_2017}. The potential AoL reduction, however, depends on whether the superposed updates can be decoded under channel-dependent interference, power allocation, and successive interference cancellation (SIC). Therefore, TWI-aware AoL minimization must jointly account for camera scheduling and power control.

In this paper, we investigate adaptive TWI selection for joint AoL-resource minimization in semantic multi-camera wireless perception under limited uplink resources.
Selecting the TWI creates a tradeoff: a longer window provides more opportunities to refresh BS-side latent beliefs, but may increase the commit-time AoL of early decoded latent representations, prolong prediction during update unavailability, and incur additional transmission and encoder costs.
To address this problem, we develop a framework that selects the TWI at the task level based on AoL-resource cost and prediction uncertainty and performs slot-level resource control within the selected TWI.
The main contributions of this paper are summarized as follows:
\begin{itemize}
    \item We introduce an AoL-aware latent-belief model for semantic multi-camera wireless perception systems. To strike a balance between AoL and resource cost, the model exploits cross-view camera correlation and prediction uncertainty to maintain BS-side latent beliefs while avoiding excessive camera scheduling. 
    
    \item 
    We formulate a joint AoL-resource minimization problem that couples task-level adaptive TWI selection with slot-level encoder selection, camera scheduling, and NOMA power allocation within the selected TWI, subject to prediction-reliability constraints. The formulation captures the tradeoff between update opportunities before commitment and the resulting AoL-resource cost, while adaptive visual encoding avoids unnecessary camera-side computation.

    \item 
    We propose a correlation-aware latent prediction algorithm for AoL minimization (CoLA) with task-level TWI selection and slot-level resource control. Lyapunov optimization selects the TWI based on prediction uncertainty and AoL-resource cost, while PPO controls camera scheduling, encoder selection, and NOMA power allocation within the selected TWI.
    
    \item 
    Numerical results characterize adaptive TWI selection under limited radio resources and camera outages, together with the tradeoffs induced by fixed TWI policies.
    CoLA achieves the best AoL-resource tradeoff among the benchmarks while maintaining prediction reliability, with larger gains under prolonged and severe camera outages.
    The results further assess the contributions of camera correlation, adaptive encoding, and NOMA-based access.
\end{itemize}

The rest of the paper is organized as follows. Section \ref{Sec:Related_Works} reviews related work. Section \ref{Sec:System_Model} introduces the multi-camera wireless perception model. Section \ref{Sec:Problem_Formulation} formulates the joint AoL-resource minimization problem. Section \ref{Sec:Solution_Approach} presents the proposed CoLA algorithm. Section \ref{Sec:Numerical_Results} provides simulation results, and Section \ref{Sec:Conclusion} concludes the paper.

\section{Related Works} \label{Sec:Related_Works}
Information freshness has been widely used to couple wireless resource allocation with the timeliness of task inputs. Beyond conventional AoI, semantic- and goal-oriented formulations relate freshness to the impact of an update on the downstream task, rather than to the generation time alone \cite{ref:freshness_to_semantics,ref:aoii_semantic,ref:aoi_semantic_twc2026,ref:aoiv_semantic_tc2025,ref:aosi_wcnc2024}. The resulting formulations provide an important step beyond packet-level timeliness, but they typically rely on information that can be explicitly evaluated after reception, such as an observed source state or a task output. Visual perception systems operating on compact latent representations present a different setting. Recent latent-dynamics methods enable resource planning using predicted representations without reconstructing raw sensory observations \cite{ref:latent_multimodal_dynamics,ref:jepa_msac_2026,ref:wireless_jepa_2026}. Multi-camera perception further introduces cross-view correlation, through which the latent representation of one camera may be predicted using contextual information from other cameras observing the same scene. Nevertheless, existing works do not model the age evolution of BS-side latent beliefs maintained across missed camera updates. Prediction is therefore mainly used to reduce communication overhead, rather than to define a freshness process for such beliefs.

Timing-aware resource management has mainly relied on AoI-based sampling and scheduling to maintain fresh task inputs under communication constraints \cite{ref:aoi_survey_yates,ref:correlated_wiener_scheduling_2025,ref:distribution_aware_aoi_lqr_2026}. Prior AoI studies have incorporated source correlation into scheduling for multi-source monitoring systems \cite{ref:correlated_sources_aoi_2022} and wireless camera networks with overlapping fields of view \cite{ref:correlated_camera_aoi_2019}. For multisensory perception, freshness alone is insufficient because task inputs generated by different sources must remain temporally coherent at decision time. TWI formalizes this requirement as a finite validity horizon for admitting sensory updates at the network--application boundary \cite{ref:twi_physical_ai}. Existing TWI studies primarily determine a network-level validity horizon from path-delay and reliability models, while abstracting away the evolution of BS-side latent beliefs and the resource decisions that refresh them. In contrast, the present work considers online TWI selection while explicitly modeling how slot-level resource decisions affect AoL accumulation and the prediction uncertainty of BS-side latent beliefs at task commitment.

Freshness-aware multiple access has recently been studied in uplink NOMA systems, where concurrent transmissions are exploited to reduce the staleness of status updates under power, decoding, and resource-allocation constraints \cite{ref:vaoi_noma_2026,ref:transformer_aoi_noma_2026}. Semantic multiple access has also been investigated for heterogeneous semantic and bit-based users, including NOMA-enabled semantic communication and age-of-incorrect-information-oriented NOMA transmission for XR applications \cite{ref:semcom_noma_2025,ref:aoii_semantic_noma_xr_2023}. Existing studies mainly model freshness as an update-level process governed by scheduling and transmission decisions. Successful decoding improves packet or version freshness, but the BS does not maintain a latent belief through prediction, and the decision time is not governed by a finite integration window. The present work differs by linking NOMA decoding success to the evolution of AoL for BS-side latent beliefs, so that scheduling and power allocation are optimized with respect to the TWI used for downstream commitment.

\begin{figure*}[!t]
    \centering
    \includegraphics[width=1\textwidth]{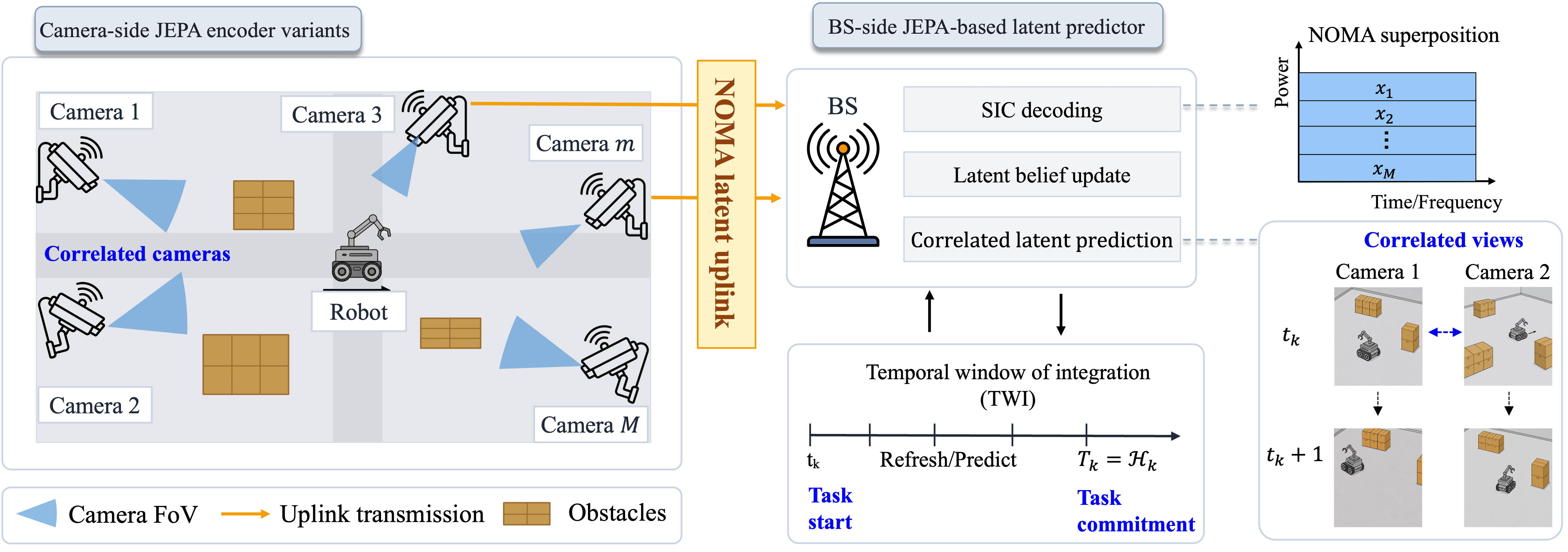}
    \caption{System model of the proposed multi-camera wireless perception framework.}
    \label{fig:system_model}
\end{figure*}

\section{System Model} \label{Sec:System_Model}
\subsection{Network Architecture}
We consider an uplink multi-camera wireless perception system, where a BS collects task-relevant latent representations from a set of correlated cameras $\mathcal{M}=\{1,\ldots,M\}$. The cameras monitor a common dynamic scene from different viewpoints, and their observations contain correlated visual information due to overlapping or complementary spatial coverage. At slot $t$, camera $m\in\mathcal{M}$ obtains an observation $X_m(t)$. Rather than uploading raw observations, a scheduled camera uses a JEPA encoder to generate a latent representation and then transmits it to the BS.
Camera-side encoding supports multiple computation--fidelity tradeoffs. Let $\mathcal{V}=\{0,\ldots,V-1\}$ denote the set of available encoder variants. When variant $v_m(t)\in\mathcal{V}$ is selected for camera $m$, the generated latent representation is
\begin{equation}
    Z_m(t)=f_{v_m(t)}\bigl(X_m(t)\bigr),
\end{equation}
where $f_v(\cdot)$ denotes the JEPA encoder corresponding to variant $v$. The output of each encoder variant is aligned to a common $d_z$-dimensional latent space, allowing the BS to process representations from different variants consistently. Encoder variant $v$ incurs a camera-side encoding cost $c_v^{\mathrm{enc}}$.

The BS serves successive instances of a downstream perception task for scene understanding and decision-making. To support these instances, the BS maintains a latent belief for each camera. Let $\bar Z_m(t)$ denote the latest successfully decoded latent representation of camera $m$ at the BS, and let $\tilde Z_m(t)$ denote the maintained BS-side latent belief used for downstream perception. Successful decoding refreshes $\bar Z_m(t)$ and the corresponding latent belief. When no fresh latent is decoded from camera $m$, the BS maintains $\tilde Z_m(t)$ through latent prediction using the previous belief for camera $m$ and information from correlated cameras. The detailed prediction and reliability model is introduced in Section~\ref{Subsec:Correlated_Prediction}.

\subsection{Communication Model}
At each slot, the BS schedules a subset of cameras to upload their latent representations. Let $o_m(t)\in\{0,1\}$ denote the scheduling decision of camera $m\in\mathcal{M}$, and let
\begin{equation}
    \mathcal{S}(t)=\{m\in\mathcal{M}:o_m(t)=1\}
\end{equation}
be the scheduled camera set. For the scheduled set $\mathcal{S}(t)$, the BS assigns transmit power $p_m(t)\ge 0$ subject to
\begin{equation}
    \sum_{m\in\mathcal{S}(t)} p_m(t) \le P_{\max},
\end{equation}
where $P_{\max}$ is the uplink power budget. The scheduled cameras transmit over a common uplink time-frequency resource using NOMA. The received signal at the BS is modeled as
\begin{equation}
    y(t)=\sum_{m\in\mathcal{S}(t)} \sqrt{p_m(t)} h_m(t) x_m(t)+n(t),
\end{equation}
where $h_m(t)$ is the uplink channel coefficient from camera $m$ to the BS, $x_m(t)$ is the unit-power signal carrying the latent representation $Z_m(t)$, and $n(t)\sim\mathcal{CN}(0,\sigma^2)$ is additive white Gaussian noise.

The BS performs uplink NOMA SIC in descending channel-gain order, following the conventional uplink NOMA decoding model \cite{ref:sic_noma_2020,ref:uplink_noma_sic}. Hence, when camera $m\in\mathcal{S}(t)$ is decoded, the signals of scheduled cameras with stronger channel gains are canceled, while those with weaker channel gains remain as interference. The residual interference for decoding camera $m$ is given by
\begin{equation}
    I_m(t) =
    \sum_{j\in\mathcal{S}(t):\, |h_j(t)|^2<|h_m(t)|^2}
    p_j(t)|h_j(t)|^2 .
\end{equation}
Accordingly, the achievable rate of camera $m$ is expressed as
\begin{equation}
    R_m(t) =
    B\log_2\left(
    1+
    \frac{
        p_m(t)|h_m(t)|^2
    }{
        I_m(t)+\sigma^2
    }
    \right),
    \quad m\in\mathcal{S}(t),
\end{equation}
where $B$ denotes the system bandwidth.
A transmitted latent representation from camera $m$ is successfully decoded if its achievable rate meets the required rate threshold $R_{\mathrm{th}}$, i.e.,
\begin{equation}
    s_m(t)=\mathbf{1}\{R_m(t)\ge R_{\mathrm{th}}\},
    \quad m\in\mathcal{S}(t).
\end{equation}
For unscheduled cameras, we set $s_m(t)=0$. The binary variable $s_m(t)$ determines whether the latent representation generated by camera $m$ at slot $t$ can refresh the corresponding BS-side latent belief.

\subsection{AoL and TWI Model}
The BS stores the latest successfully decoded latent representation of each camera. Let $G_m(t)$ denote the generation time of the stored latent representation $\bar Z_m(t)$ of camera $m$. The stored latent representation and its generation time evolve as
\begin{equation}
    \bar Z_m(t)=
    \begin{cases}
        Z_m(t), & s_m(t)=1,\\
        \bar Z_m(t-1), & s_m(t)=0,
    \end{cases}
\end{equation}
and
\begin{equation}
    G_m(t)=
    \begin{cases}
        t, & s_m(t)=1,\\
        G_m(t-1), & s_m(t)=0.
    \end{cases}
\end{equation}
The AoL of camera $m$ measures the freshness of its latest successfully decoded latent representation and is defined as
\begin{equation}
    \Delta_m(t)=t-G_m(t).
\end{equation}
The AoL evolution can be expressed by
\begin{equation}
    \Delta_m(t)=
    \begin{cases}
        0, & s_m(t)=1,\\
        \Delta_m(t-1)+1, & s_m(t)=0.
    \end{cases}
    \label{eq:aol_evolution}
\end{equation}
Hence, AoL is reset only when a fresh latent from camera $m$ is successfully decoded. If no fresh latent is decoded, the BS maintains the corresponding belief through look-ahead prediction. In this case, the BS does not update $\bar Z_m(t)$ or $G_m(t)$, and $\Delta_m(t)$ continues to increase. Although AoL follows the same temporal evolution as AoI, it tracks the age of the latest successfully decoded latent representation associated with each BS-side latent belief. The reliability of a belief maintained over consecutive look-ahead predictions is characterized by the uncertainty model in Section~\ref{Subsec:Correlated_Prediction}, following the latent-dynamics planning principle that prediction uncertainty accumulates when fresh latents are unavailable \cite{ref:latent_multimodal_dynamics}.

For the $k$-th task instance, the BS opens a temporal window of integration (TWI) at slot $t_k$ and selects a length $\tau_k\in\Omega$, where $\Omega$ is the feasible set of TWI lengths \cite{ref:twi_physical_ai}. Accordingly, $\tau_k$ determines the task-commitment slot and the available uplink communication opportunities. Since $\tau_k=0$ corresponds to no communication window and the task is served at slot $t_k$, we define the service length as
\begin{equation}
    L(\tau_k)=\max\{1,\tau_k\}.
\end{equation}
The corresponding service horizon is
\begin{equation}
    \mathcal{H}_k=\{t_k,\ldots,T_k\}, \quad T_k=t_k+L(\tau_k)-1,
\end{equation}
and the maintained BS-side latent beliefs are committed to the downstream task at slot $T_k$. The accumulated AoL over $\mathcal{H}_k$ is defined as
\begin{equation}
    C_k^{\mathrm{AoL}} =
    \sum_{t\in\mathcal{H}_k}
    \frac{1}{M}\sum_{m\in\mathcal{M}}\Delta_m(t).
    \label{eq:task_aol}
\end{equation}
A larger TWI provides more opportunities to decode fresh latents, but may also extend the prediction duration during update unavailability and accumulate additional AoL before commitment.

\subsection{Correlated Latent Prediction and Uncertainty Model}
\label{Subsec:Correlated_Prediction}
Cross-view redundancy enables the BS to exploit decoded latent representations and maintained beliefs from correlated cameras when a fresh latent from the target camera is unavailable. Let $\boldsymbol{\xi}_m(t)$ collect the AoL, prediction depth, belief mode, and availability metadata for camera $m$. Prediction depth is the number of consecutive predictions since the latest fresh decoding. Belief mode indicates whether the current belief is obtained through fresh decoding, prediction from the camera's previous belief, or prediction using correlated cameras. For camera $m$, let $\mathcal{D}_m(t-1)$ denote the tuple comprising its maintained belief and metadata available at the BS before slot $t$. Accordingly,
\begin{equation}
    \mathcal{D}_m(t-1)=\left(\tilde Z_m(t-1),\boldsymbol{\xi}_m(t-1)\right).
\end{equation}
After the decoding results at slot $t$ are obtained, the BS forms the information set $\mathcal{C}_m(t)$ for camera $m$. In $\mathcal{C}_m(t)$, $\boldsymbol{\xi}_j(t)$ is paired with $\bar Z_j(t)$ when $s_j(t)=1$ and with $\tilde Z_j(t-1)$ otherwise, where $j$ indexes a camera correlated with $m$. Formally,
\begin{equation}
    \begin{aligned}
    \mathcal{C}_m(t) = \Bigl\{&
    \left(
        s_j(t)\bar Z_j(t) +
        \bigl(1-s_j(t)\bigr)\tilde Z_j(t-1),
        \boldsymbol{\xi}_j(t)
    \right)
    \,\Bigm|\\
    &j\in\mathcal{M}\setminus\{m\}
    \Bigr\}.
    \end{aligned}
\end{equation}
When no fresh latent is decoded from camera $m$, the BS performs look-ahead latent prediction as
\begin{equation}
    \hat Z_m(t)=q_m\bigl(\mathcal{D}_m(t-1),\mathcal{C}_m(t)\bigr).
\end{equation}
Here, $q_m(\cdot)$ denotes an offline-trained BS-side look-ahead latent predictor that maps $\mathcal{D}_m(t-1)$ and $\mathcal{C}_m(t)$ into a predicted latent representation. Training samples are constructed from synchronized camera frames. The loss combines latent-space MSE to the corresponding fresh latent representation with a penalty when prediction increases the error of the input estimate. The maintained BS-side latent belief is then updated according to
\begin{equation}
    \tilde Z_m(t)=
    \begin{cases}
        \bar Z_m(t), & s_m(t)=1,\\
        \hat Z_m(t), & s_m(t)=0.
    \end{cases}
\end{equation}
Thus, fresh decoding and look-ahead prediction play different roles: the former refreshes the stored latent representation, while the latter maintains the BS-side belief when no fresh latent is decoded.

Consecutive look-ahead predictions can degrade the reliability of the maintained latent belief when fresh latents are unavailable. Let $u_m(t)$ denote the uncertainty associated with $\tilde Z_m(t)$. This quantity serves as a reliability surrogate for downstream commitment. During offline calibration, the normalized latent-space substitution error is defined as
\begin{equation}
    e_m^{\mathrm{sub}}(t) =
    \frac{1}{d_z}
    \left\|
    \tilde Z_m(t)-\bar Z_m^{\mathrm{fr}}(t)
    \right\|^2,
\end{equation}
where $d_z$ is the latent dimension and $\bar Z_m^{\mathrm{fr}}(t)$ denotes the fresh-reference latent representation generated from camera $m$'s observation at slot $t$ during offline calibration. Unlike $\bar Z_m(t)$, this reference does not depend on scheduling or successful decoding and is unavailable during online operation. The substitution error uses the same latent-space MSE as the predictor objective and serves as the target for uncertainty calibration. To estimate uncertainty, we use an ensemble of $E$ conditional predictors that share the architecture and inputs of the look-ahead predictor $q_m(\cdot)$ in (17). Each member is independently trained on a bootstrap resample of the same training set using the same objective. For the inputs $\mathcal{D}_m(t-1)$ and $\mathcal{C}_m(t)$, the $e$th predictor $q_m^{(e)}(\cdot)$ generates
\begin{equation}
    \hat Z_m^{(e)}(t)=q_m^{(e)}\bigl(\mathcal{D}_m(t-1),\mathcal{C}_m(t)\bigr),
    \quad e=1,\ldots,E.
\end{equation}
The ensemble average is
\begin{equation}
    \hat Z_m^{\mathrm{avg}}(t)=\frac{1}{E}\sum_{e=1}^{E}\hat Z_m^{(e)}(t).
\end{equation}
The variance across ensemble outputs in each latent dimension, normalized squared deviation from their average, and normalized squared difference between pairs of outputs are defined as
\begin{equation}
    \begin{aligned}
    v_{m,i}(t)&=\frac{1}{E}\sum_{e=1}^{E}
    \left([\hat Z_m^{(e)}(t)]_i-[\hat Z_m^{\mathrm{avg}}(t)]_i\right)^2,\\
    a_{m,e}(t)&=\frac{1}{d_z}\left\|\hat Z_m^{(e)}(t)-\hat Z_m^{\mathrm{avg}}(t)\right\|^2,\\
    p_{m,e,e'}(t)&=\frac{1}{d_z}\left\|\hat Z_m^{(e)}(t)-\hat Z_m^{(e')}(t)\right\|^2.
    \end{aligned}
\end{equation}
Here, $i=1,\ldots,d_z$ indexes the latent components. The condition $1\le e<e'\le E$ ensures that each predictor pair is considered once. The six ensemble-disagreement features are the mean and maximum of $v_{m,i}(t)$ over $i$, $a_{m,e}(t)$ over $e$, and $p_{m,e,e'}(t)$ over $e<e'$. The uncertainty feature vector $\phi_m(t)$ combines these ensemble-disagreement features with quantities derived from the target-camera metadata $\boldsymbol{\xi}_m(t)$ and aggregate measures of the correlated-camera inputs in $\mathcal{C}_m(t)$, including their number, AoL, prediction depth, belief mode, and availability. We train a heteroscedastic error predictor $g_\psi(\cdot)$ using $\phi_m(t)$ as input and $e_m^{\mathrm{sub}}(t)$ as the target, with Gaussian negative log-likelihood. It outputs
\begin{equation}
    \left(\hat\mu_m^{\mathrm{sub}}(t),\hat\sigma_{m,\mathrm{raw}}^{\mathrm{sub}}(t)\right)
    =g_\psi\bigl(\phi_m(t)\bigr).
\end{equation}
The raw standard deviation is then rescaled by an offline-fitted factor selected according to whether correlated-camera inputs are considered and their number, yielding $\hat\sigma_{m,\mathrm{cal}}^{\mathrm{sub}}(t)$. By Cantelli's one-sided inequality~\cite{ref:cantelli_1928}, the online uncertainty surrogate is then defined as
\begin{equation}
    u_m(t) = 
    \hat\mu_m^{\mathrm{sub}}(t) +
    \hat\sigma_{m,\mathrm{cal}}^{\mathrm{sub}}(t)
    \sqrt{\frac{1-\delta_m}{\delta_m}},
    \label{eq:uncertainty_surrogate}
\end{equation}
where $\delta_m\in(0,1)$ specifies the one-sided uncertainty level. Given a target uncertainty level $u_m^{\mathrm{tar}}$, the maintained belief of camera $m$ is admissible at slot $t$ if
\begin{equation}
    u_m(t)\le u_m^{\mathrm{tar}}.
\end{equation}
Here, $u_m^{\mathrm{tar}}$ denotes a fixed reliability threshold derived offline on the same substitution-error scale and is kept fixed during TWI selection and slot-level resource control.
AoI alone cannot distinguish BS-side latent beliefs with the same age but different reliability. The pair $(\Delta_m(t),u_m(t))$ captures this distinction because $\Delta_m(t)=\Delta_j(t)$ does not generally imply $u_m(t)=u_j(t)$ when the beliefs are maintained from different prediction histories or correlated-camera information.

\subsection{Problem Formulation}
\label{Sec:Problem_Formulation}
The encoder variants, latent predictors, and uncertainty model are trained offline and treated as fixed during online optimization.
The objective is to select the TWI for each task instance and jointly optimize camera scheduling, encoder-variant selection, and transmit-power allocation over the resulting communication window. Scheduling more cameras or selecting higher-cost encoder variants can improve the freshness and reliability of BS-side latent beliefs, at the cost of higher wireless and camera-side resource consumption.

Let $\mathcal{K}=\{1,\ldots,K\}$ index $K$ sequential instances of the downstream perception task. For each instance $k\in\mathcal{K}$, the AoL cost $C_k^{\mathrm{AoL}}$ is defined in \eqref{eq:task_aol}. The communication window associated with $\tau_k$ is defined as
\begin{equation}
    \mathcal{T}_k =
    \begin{cases}
        \emptyset, & \tau_k=0,\\
        \mathcal{H}_k, & \tau_k>0.
    \end{cases}
\end{equation}
The transmission cost and camera-side encoding cost over $\mathcal{T}_k$ are given by
\begin{equation}
    C_k^{\mathrm{tx}} =
    \sum_{t\in\mathcal{T}_k}
    \sum_{m\in\mathcal{M}}
    p_m(t),
\end{equation}
and
\begin{equation}
    C_k^{\mathrm{enc}} =
    \sum_{t\in\mathcal{T}_k}
    \sum_{m\in\mathcal{M}}
    o_m(t)c_{v_m(t)}^{\mathrm{enc}},
\end{equation}
respectively. The total resource cost is then defined as
\begin{equation}
    C_k^{\mathrm{res}} =
    C_k^{\mathrm{tx}} +
    \eta_{\mathrm{enc}} C_k^{\mathrm{enc}},
\end{equation}
where $\eta_{\mathrm{enc}}$ controls the relative importance of camera-side encoding cost.

Let $\omega_{\mathrm{AoL}}$ and $\omega_{\mathrm{res}}$ denote the nonnegative weights assigned to the AoL and resource costs, respectively. The finite-horizon joint optimization problem is formulated as
\begin{subequations}\label{P1}
\begin{align}
    \mathbf{P1}:\quad
    \min_{\mathcal{X}}\quad
    & \frac{1}{K}\sum_{k\in\mathcal{K}}
    \mathbb{E}\!\left[
    \omega_{\mathrm{AoL}} C_k^{\mathrm{AoL}} + \omega_{\mathrm{res}} C_k^{\mathrm{res}}
    \right] \label{P1a}\\
    \text{s.t.}\quad
    & u_m(T_k)\le u_m^{\mathrm{tar}}, \forall (m,k)\in\mathcal{M}\times\mathcal{K},
    \label{P1b} \\
    & \sum_{m\in\mathcal{M}} p_m(t)\le P_{\max}, \forall k\in\mathcal{K},\; t\in\mathcal{T}_k,
    \label{P1c}\\
    & \tau_k\in\Omega, \forall k\in\mathcal{K}.
    \label{P1d}
\end{align}
\end{subequations}
In optimization problem~\eqref{P1}, the decision-variable set is $ \mathcal{X} = \{\tau_k:k\in\mathcal{K}\} \cup \{o_m(t),v_m(t),p_m(t): (m,k)\in\mathcal{M}\times\mathcal{K},\,t\in\mathcal{T}_k\}.$ When $o_m(t)=1$, the encoder variant satisfies $v_m(t)\in\mathcal{V}$; when $o_m(t)=0$, no encoder variant is applied and $p_m(t)=0$. Constraint \eqref{P1b} ensures that the BS-side latent beliefs committed to downstream perception satisfy the target uncertainty level. Constraint \eqref{P1c} imposes the per-slot uplink power budget during the communication window of each task. Constraint \eqref{P1d} restricts the TWI length to the feasible set. The expectation in \eqref{P1} accounts for the stochastic channel evolution and the decoding results, which determine whether latent beliefs are refreshed or maintained by prediction.

\begin{proposition}
Problem $\mathbf{P1}$ is NP-hard.
\end{proposition}

\begin{figure}[!t]
    \centering
    \includegraphics[width=3.3in]{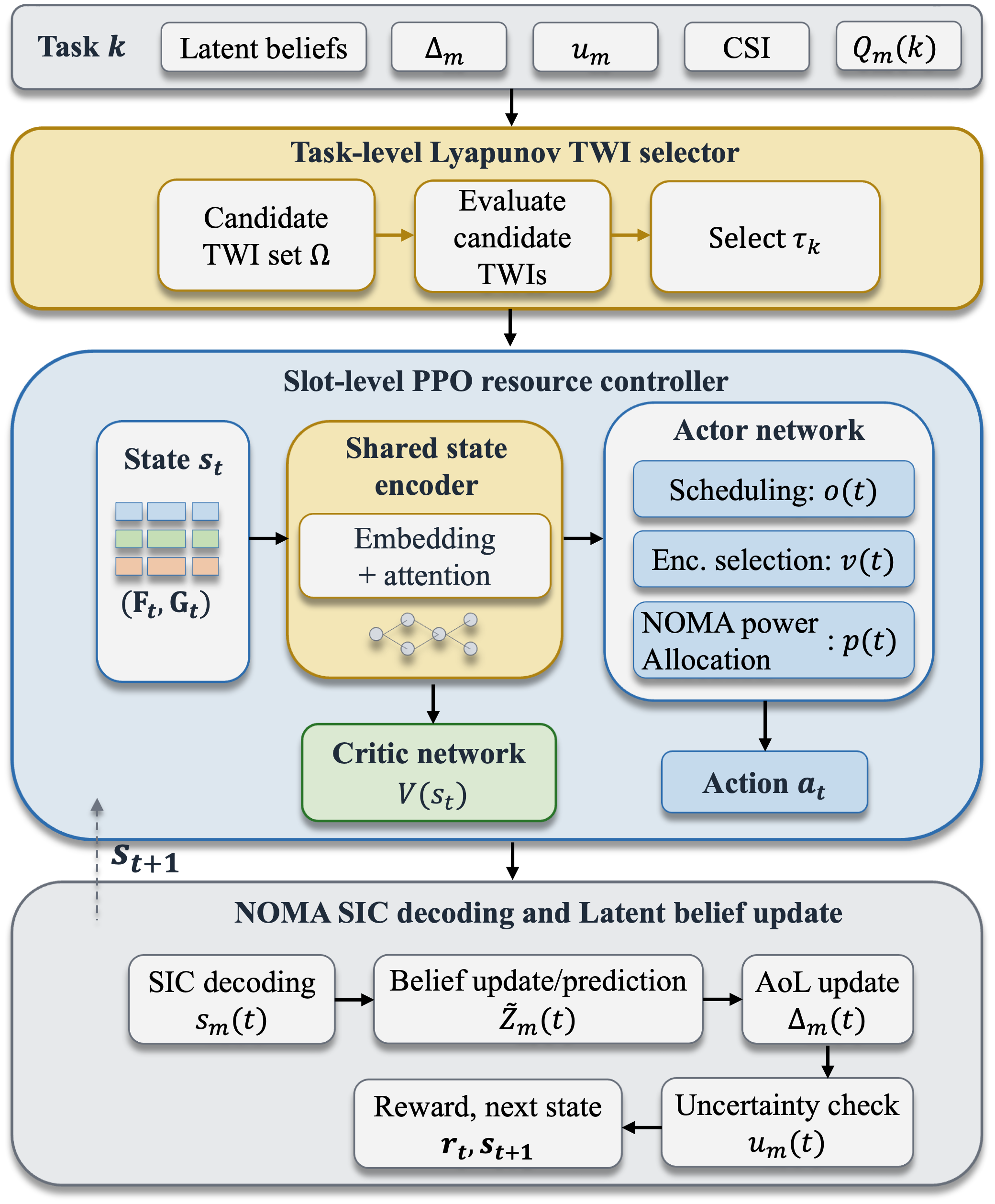}
    \caption{Architecture of the proposed CoLA algorithm with task-level TWI selection and slot-level PPO resource control.}
    \label{fig:CoLA_framework}
\end{figure}

{\renewcommand{\IEEEproofindentspace}{0pt}
\begin{IEEEproof}
We prove the result by reduction from the Set Cover problem \cite{Karp1972}. Consider an arbitrary Set Cover instance with a universe $\mathcal{U}=\{e_1,\ldots,e_N\}$, a collection of subsets $\{\mathcal{S}_1,\ldots,\mathcal{S}_L\}$, and a budget $B$. We construct a special case of $\mathbf{P1}$ with a single task instance, i.e., $K=1$, and a fixed TWI length, i.e., $\Omega=\{1\}$. The wireless channels are deterministic, and the AoL component in the objective is removed by setting $\omega_{\mathrm{AoL}}=0$.
For each subset $\mathcal{S}_j$, we create one schedulable camera. Scheduling this camera corresponds to selecting subset $\mathcal{S}_j$ and incurs one unit of resource cost. For each element $e_i\in\mathcal{U}$, we create one target latent belief whose uncertainty constraint must be satisfied at the task-commitment slot. The uncertainty model is set such that $u_i(T_1)\le u_i^{\mathrm{tar}}$ holds if and only if at least one scheduled camera corresponding to a subset $\mathcal{S}_j$ with $e_i\in\mathcal{S}_j$ provides a decoded latent representation for predicting that belief. Otherwise, the uncertainty constraint is violated.
Under this construction, satisfying all uncertainty constraints in \eqref{P1b} is equivalent to covering every element in $\mathcal{U}$. Moreover, the resource cost in the objective equals the number of selected subsets. Hence, there exists a feasible solution to this special case of $\mathbf{P1}$ with objective value no larger than $B$ if and only if the original Set Cover instance admits a cover of size at most $B$.
The construction introduces $L$ schedulable cameras and $N$ target latent beliefs, with the uncertainty conditions specified directly from the incidences $e_i\in\mathcal{S}_j$. It can therefore be completed in time polynomial in the size of the Set Cover instance. Thus, Set Cover polynomially reduces to the decision version of $\mathbf{P1}$, and the optimization problem $\mathbf{P1}$ is NP-hard.
\end{IEEEproof}}

\section{Solution Approach: CoLA}
\label{Sec:Solution_Approach}
Since $\mathbf{P1}$ is NP-hard, directly optimizing $\mathcal{X}$ over all tasks is computationally intractable. We decompose the problem according to the natural task-level/communication-slot hierarchy of the system. At the task level, CoLA selects the TWI length by balancing the estimated AoL-resource cost and the reliability violation at the task-commitment slot. Within the selected TWI, the slot-level controller determines camera scheduling, encoder-variant selection, and NOMA power allocation. 

\subsection{Slot-Level Resource Controller}
\label{Subsec:Slot_Level_Control}
Within $\mathcal{T}_k$, we formulate slot-level resource control as a finite-horizon MDP, where the SIC decoding indicators determine whether each BS-side latent belief is refreshed by a decoded latent representation or maintained through look-ahead prediction.
\begin{itemize}
    \item \textbf{State:} The state $\mathbf{s}_t$ consists of a per-camera feature matrix $\mathbf{F}_t\in\mathbb{R}^{M\times d_{\mathrm{f}}}$ and a global feature vector $\mathbf{G}_t\in\mathbb{R}^{d_{\mathrm{g}}}$, where $d_{\mathrm{f}}$ and $d_{\mathrm{g}}$ denote the corresponding feature dimensions. For each camera $m$, $\mathbf{F}_t$ includes the AoL $\Delta_m(t)$, prediction uncertainty $u_m(t)$, channel gain $|h_m(t)|^2$, the correlated-camera context feature, and the encoding cost associated with the maintained belief. $\mathbf{G}_t$ includes the selected TWI length $\tau_k$, the remaining slots before $T_k$, the transmission and decoding status accumulated within $\mathcal{T}_k$, and the encoder-cost profile $\{c_v^{\mathrm{enc}}\}_{v\in\mathcal{V}}$.

    \item \textbf{Action:} The action $\mathbf{a}_t=(\mathbf{o}(t),\mathbf{v}(t),\mathbf{p}(t))$ consists of the scheduling, encoder-variant selection, and power-allocation decisions at slot $t$, where $\mathbf{o}(t)$, $\mathbf{v}(t)$, and $\mathbf{p}(t)$ collect $\{o_m(t)\}_{m\in\mathcal{M}}$, $\{v_m(t)\}_{m\in\mathcal{M}}$, and $\{p_m(t)\}_{m\in\mathcal{M}}$, respectively. The action is generated hierarchically: the controller first determines the scheduled camera set, and then determines the encoder-variant and NOMA power-allocation decisions conditioned on the scheduled set.
    
    \item \textbf{Reward:} The reward $\mathbf{r}_t$ is designed to penalize the AoL and resource costs induced by the current slot-level action. A commitment-slot penalty is applied when the BS-side latent beliefs violate the uncertainty requirement at $T_k$. Let $c^{\mathrm{AoL}}(t) = \frac{1}{M}\sum_{m\in\mathcal{M}}\Delta_m(t)$ denote the average AoL at slot $t$, and let $c^{\mathrm{res}}(t) = \sum_{m\in\mathcal{M}}p_m(t) + \eta_{\mathrm{enc}}\sum_{m\in\mathcal{M}}o_m(t)c_{v_m(t)}^{\mathrm{enc}}$ denote the slot-level resource cost. Let $\chi_k(t)=\mathbf{1}\{t=T_k\}$ indicate whether slot $t$ is the task-commitment slot. The reward is given by
    \begin{equation}
        \label{Eq:Slot_Reward}
        \mathbf{r}_t = - \omega_{\mathrm{AoL}}c^{\mathrm{AoL}}(t) -\omega_{\mathrm{res}}c^{\mathrm{res}}(t) - \chi_k(t)\lambda_{\mathrm{rel}}D_k^{\mathrm{rel}},
    \end{equation}
    where $\lambda_{\mathrm{rel}}\ge 0$ is the reliability penalty weight and $D_k^{\mathrm{rel}}$ is the uncertainty-violation cost. The last term corresponds to a Lagrangian relaxation of the uncertainty constraint in \eqref{P1b}. Since this constraint is imposed at task commitment, the reliability penalty is activated only when $t=T_k$. Specifically, we define
    \begin{equation}
        \label{Eq:Rel_Violation_Cost}
        D_k^{\mathrm{rel}} =
        \frac{V_k^{\mathrm{rel}}}{1+V_k^{\mathrm{rel}}},
        \quad
        V_k^{\mathrm{rel}} =
        \max_{m\in\mathcal{M}}
        \frac{\left[u_m(T_k)-u_m^{\mathrm{tar}}\right]^+}
        {u_m^{\mathrm{tar}}+\epsilon},
    \end{equation}
    where $\epsilon>0$ is a small constant. Here, $[x]^+=\max\{x,0\}$ denotes the positive-part operator.
\end{itemize}

\begin{algorithm}[t]
\caption{Communication-Slot-Level Controller Algorithm}
\label{Alg:Inner_PPO}
\begin{algorithmic}[1]
\STATE \textbf{Input:} Task sequences, initial actor $\pi_{\theta}$, and critic $V_{\phi}$.
\STATE \textbf{Initialize:} Rollout buffer $\mathcal{B}$ and reliability weight $\lambda_{\mathrm{rel}}$.
\FOR{each training episode}
    \STATE Clear $\mathcal{B}$.
    \FOR{each task $k$ in the episode}
        \FOR{each communication slot $t\in\mathcal{T}_k$}
            \STATE Construct $\mathbf{s}_t=(\mathbf{F}_t,\mathbf{G}_t)$.
            \STATE Sample $\mathbf{a}_t$ from $\pi_{\theta}(\cdot|\mathbf{s}_t)$ using \eqref{Eq:Policy_Factorization}.
            \STATE Execute $\mathbf{a}_t$, compute $\mathbf{r}_t$ using \eqref{Eq:Slot_Reward}, and store the transition in $\mathcal{B}$.
        \ENDFOR
    \ENDFOR
    \STATE Compute $\{\hat R_t,\hat A_t\}$ from $\mathcal{B}$.
    \FOR{each PPO epoch}
        \FOR{each mini-batch sampled from $\mathcal{B}$}
            \STATE Compute $\rho_t(\theta)$ using \eqref{Eq:PPO_Ratio}.
            \STATE Update $\theta$ and $\phi$ by minimizing \eqref{Eq:PPO_Total_Loss}.
        \ENDFOR
    \ENDFOR
    \STATE Update $\lambda_{\mathrm{rel}}$ using task-commitment violation statistics.
\ENDFOR
\end{algorithmic}
\end{algorithm}

We adopt an actor-critic architecture for the slot-level controller. To handle the high-dimensional heterogeneous slot state, a shared state encoder maps $\mathbf{s}_t$ into a compact representation via feature embedding and an attention mechanism. The actor network generates the hierarchical resource-control action by first scheduling cameras and then selecting encoder variants and NOMA transmit powers for the scheduled cameras. The critic network uses the same encoded representation to estimate the value function $V_{\phi}(\mathbf{s}_t)$. The resulting hierarchical policy is represented by the joint log-probability as
\begin{equation}
\label{Eq:Policy_Factorization}
\begin{aligned}
\log \pi_{\theta}(\mathbf{a}_t|\mathbf{s}_t) ={}&
\log \pi_{\theta}^{\mathrm{sch}}(\mathbf{o}(t)|\mathbf{s}_t) +
\log \pi_{\theta}^{\mathrm{enc}}(\mathbf{v}(t)|\mathbf{o}(t),\mathbf{s}_t) \\
&+\log \pi_{\theta}^{\mathrm{pow}}(\mathbf{p}(t)|\mathbf{o}(t),\mathbf{s}_t),
\end{aligned}
\end{equation}
where $\pi_{\theta}^{\mathrm{sch}}$, $\pi_{\theta}^{\mathrm{enc}}$, and $\pi_{\theta}^{\mathrm{pow}}$ denote the conditional scheduling, encoder-variant selection, and power-allocation distributions induced by the actor, respectively. The slot-level policy is trained using the proximal policy optimization (PPO) framework. Let
\begin{equation}
    \label{Eq:PPO_Ratio}
    \rho_t(\theta) =
    \exp\left(
    \log \pi_{\theta}(\mathbf{a}_t|\mathbf{s}_t) -
    \log \pi_{\theta_{\mathrm{old}}}(\mathbf{a}_t|\mathbf{s}_t)
    \right)
\end{equation}
denote the likelihood ratio between the current and previous policies. The clipped surrogate objective is given by
\begin{equation}
    \label{Eq:PPO_Objective}
    \begin{aligned}
        \mathcal{L}^{\mathrm{PPO}}(\theta) = 
        & \mathbb{E}_t \Big[
        \min\big(
        \rho_t(\theta)\hat A_t, \\
        &\operatorname{clip}\big(\rho_t(\theta),1-\epsilon_{\mathrm{clip}},1+\epsilon_{\mathrm{clip}}\big)\hat A_t
        \big)
        \Big],
    \end{aligned}
\end{equation}
where $\epsilon_{\mathrm{clip}}>0$ is the clipping parameter and $\hat A_t$ is the advantage estimate computed from the Lagrangian-augmented reward. Let $V_{\phi}(\mathbf{s}_t)$ denote the value function and let $\hat R_t$ denote the value target. The critic loss is defined as
\begin{equation}
    \label{Eq:Value_Loss}
    \mathcal{L}^{\mathrm{V}}(\phi) =
    \mathbb{E}_t
    \left[
    \left(
    V_{\phi}(\mathbf{s}_t)-\hat R_t
    \right)^2
    \right].
\end{equation}
The entropy regularization, denoted by $\mathcal{L}^{\mathrm{H}}(\theta)$, is the expected sum of the entropies of the scheduling, encoder-variant-selection, and power-allocation distributions.
The actor and critic parameters are updated by minimizing
\begin{equation}
    \label{Eq:PPO_Total_Loss}
    \mathcal{L}^{\mathrm{total}}(\theta,\phi) = -\mathcal{L}^{\mathrm{PPO}}(\theta) + c_{\mathrm{V}}\mathcal{L}^{\mathrm{V}}(\phi) - c_{\mathrm{H}}\mathcal{L}^{\mathrm{H}}(\theta),
\end{equation}
where $c_{\mathrm{V}}$ and $c_{\mathrm{H}}$ are nonnegative weights. Algorithm~\ref{Alg:Inner_PPO} summarizes the PPO training procedure for the communication-slot-level controller.

\begin{algorithm}[t]
\caption{CoLA: Correlation-aware Latent Prediction for AoL Minimization}
\label{Alg:CoLA}
\begin{algorithmic}[1]
\STATE \textbf{Input:} Candidate TWI set $\Omega$ and the trained slot-level controller $\pi_{\theta}$.
\STATE \textbf{Initialize:} $Q_m(1)=0$, $\forall m\in\mathcal{M}$.
\FOR{each task $k$}
    \STATE Observe the current latent beliefs, channels, and queues $\{Q_m(k)\}_{m\in\mathcal{M}}$.
    \FOR{each candidate $\tau\in\Omega$}
        \STATE Estimate the score in \eqref{Eq:TWI_Score} under $\pi_{\theta}$.
    \ENDFOR
    \STATE Select $\tau_k$ according to \eqref{Eq:TWI_Score}.
    \STATE Execute the selected TWI using $\pi_{\theta}$ and update the BS-side latent beliefs.
    \STATE Compute $\{\tilde d_m(k)\}_{m\in\mathcal{M}}$ and update the queues according to \eqref{Eq:Queue_Update}.
\ENDFOR
\end{algorithmic}
\end{algorithm}

\subsection{TWI Selection via Lyapunov Optimization}
\label{Subsec:TWI_Selection}
After obtaining the slot-level resource controller, CoLA determines the TWI length through task-level Lyapunov optimization. The uncertainty constraint in \eqref{P1b} is imposed on the BS-side latent beliefs committed to downstream inference and thus couples the reliability performance across successive task instances. To track this constraint violation pressure, we maintain a virtual reliability queue for each camera $m$. Define the normalized reliability residual at the commitment slot of task $k$ as
\begin{equation}
    \label{Eq:Rel_Residual}
    d_m(k) =
    \frac{u_m(T_k)-u_m^{\mathrm{tar}}}{u_m^{\mathrm{tar}}+\epsilon},
\end{equation}
where $\epsilon>0$ is a small constant. Let $\tilde d_m(k)$ denote a bounded version of $d_m(k)$. The virtual reliability queue is updated after task $k$ as
\begin{equation}
    \label{Eq:Queue_Update}
    Q_m(k+1) =
    \left[
    Q_m(k)+\eta_Q\tilde d_m(k)
    \right]^+,
    \quad m\in\mathcal{M},
\end{equation}
where $\eta_Q$ is the queue update stepsize.

We evaluate each candidate $\tau\in\Omega$ by estimating the resulting AoL-resource cost and uncertainty-constraint violation from the current system state. $\tau=0$ represents a no-transmission decision, which avoids transmission and encoder costs but relies on the maintained BS-side latent beliefs at commitment. Define $G_k\triangleq\sum_{m\in\mathcal{M}}[\tilde d_m(k)]^+$ as the positive reliability-violation cost. Let $\widehat{C}_{k}^{\mathrm{AoL}}(\tau)$, $\widehat{C}_{k}^{\mathrm{res}}(\tau)$, $\widehat{\tilde d}_{m,k}(\tau)$, and $\widehat{G}_{k}(\tau)$ denote the candidate-dependent estimates of $C_k^{\mathrm{AoL}}$, $C_k^{\mathrm{res}}$, $\tilde d_m(k)$, and $G_k$, respectively.
The drift-plus-penalty TWI selection rule is given by
\begin{equation}
\label{Eq:TWI_Score}
\begin{aligned}
    \tau_k = 
    \arg\min_{\tau\in\Omega}
    \Big[
    &V_{\mathrm{L}}
    \big(
    \omega_{\mathrm{AoL}}\widehat{C}_{k}^{\mathrm{AoL}}(\tau) + 
    \omega_{\mathrm{res}}\widehat{C}_{k}^{\mathrm{res}}(\tau)
    \big) \\
    &+\sum_{m\in\mathcal{M}}Q_m(k)\widehat{\tilde d}_{m,k}(\tau) + 
    \lambda_0\widehat{G}_{k}(\tau)
    \Big],
\end{aligned}
\end{equation}
where $V_{\mathrm{L}}>0$ controls the weight of the AoL-resource objective, and $\lambda_0\ge 0$ is a base penalty for predicted reliability violation. After selecting $\tau_k$, CoLA applies the controller over $\mathcal{T}_k$, computes $\{\tilde d_m(k)\}_{m\in\mathcal{M}}$ from $\{u_m(T_k)\}_{m\in\mathcal{M}}$, and updates the virtual queues according to \eqref{Eq:Queue_Update}. Algorithm~\ref{Alg:CoLA} summarizes the overall CoLA procedure.

\subsection{Complexity Analysis}
\label{Subsec:Complexity_Analysis}
The slot-level resource controller in Algorithm~\ref{Alg:Inner_PPO} is trained offline. Hence, the online computational complexity comes from the task-level TWI selection and the controller execution in Algorithm~\ref{Alg:CoLA}. Let $\mathcal{C}_{\mathrm{in}}$ denote the computational complexity of one communication-slot update, which includes forming the slot state, inferring the resource-control action, computing the SIC decoding indicators, and updating the BS-side latent belief and uncertainties.

For each task $k$, Algorithm~\ref{Alg:CoLA} first evaluates the score terms in \eqref{Eq:TWI_Score} for all candidate TWI lengths. For each positive candidate $\tau>0$, the corresponding estimate involves $\tau$ communication-slot updates. The case $\tau=0$ does not invoke the slot-level controller because it represents no communication window. Hence, the candidate-evaluation cost is $\mathcal{O}((\sum_{\tau\in\Omega\setminus\{0\}}\tau)\mathcal{C}_{\mathrm{in}})$. Once the candidate estimates are obtained, forming \eqref{Eq:TWI_Score} for all candidates requires $\mathcal{O}(|\Omega|M)$ operations because the queue term sums over the $M$ cameras for each candidate. The minimization over $\Omega$ and the queue update in \eqref{Eq:Queue_Update} are absorbed into the same order. Executing the selected TWI adds $\mathcal{O}(\tau_k\mathcal{C}_{\mathrm{in}})$ complexity. Combining the candidate evaluation, score formation, selected-TWI execution, and queue update, the per-task online complexity of Algorithm~\ref{Alg:CoLA} is bounded by $\mathcal{O}(|\Omega|(\tau_{\max}\mathcal{C}_{\mathrm{in}}+M))$, where $\tau_{\max}=\max_{\tau\in\Omega}\tau$. For $K$ task instances, the overall online complexity is $\mathcal{O}(K|\Omega|(\tau_{\max}\mathcal{C}_{\mathrm{in}}+M))$.

\begin{figure*}[!t]
    \centering
    \includegraphics[width=0.85\textwidth]{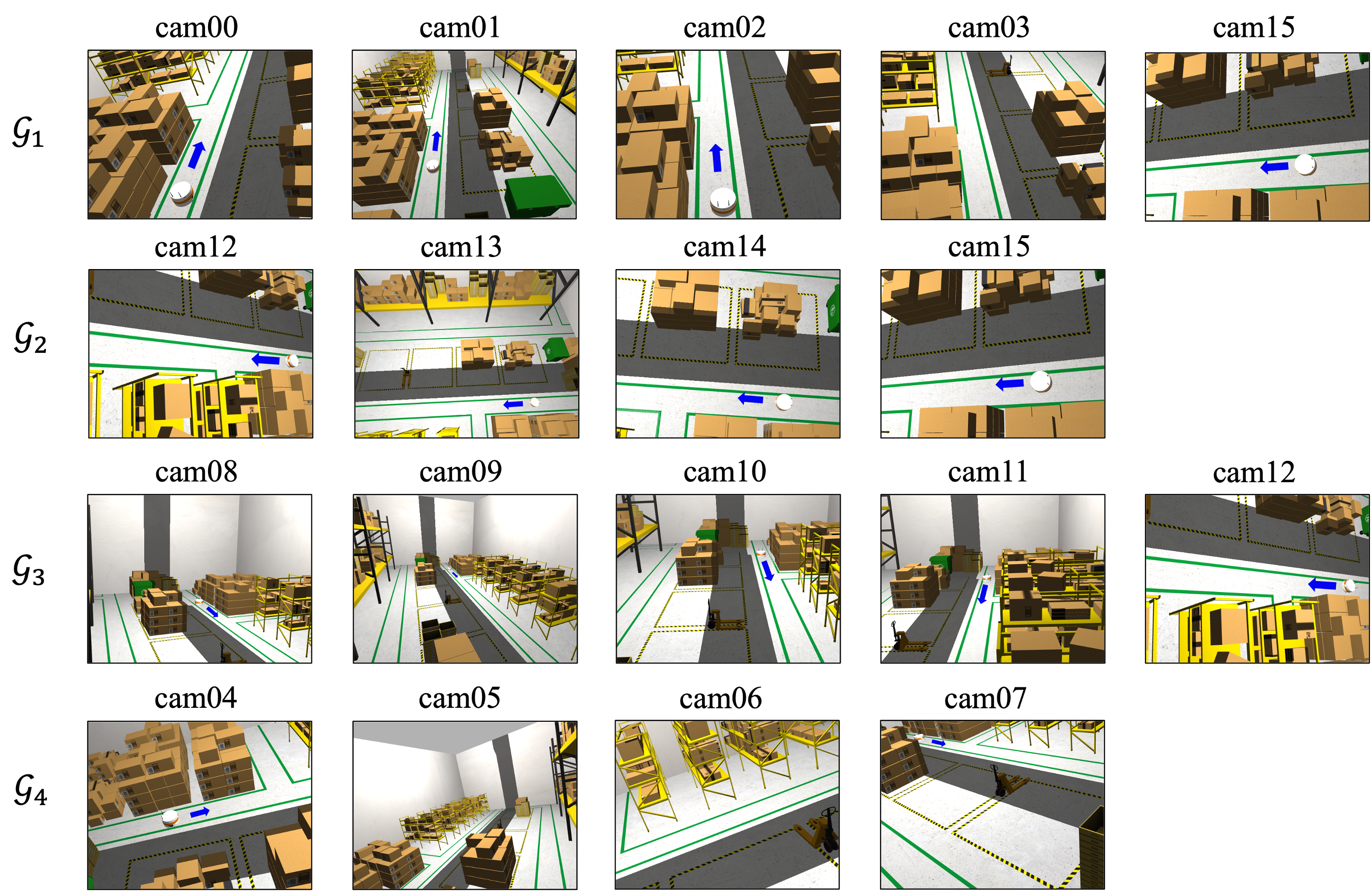}
    \caption{Representative group-wise synchronized camera views from the Warehouse MultiCam RF Dataset~\cite{ref:warehouse_dataset_zenodo}. All views are taken from frame \textit{64}, and the blue arrows indicate the robot moving direction. Rows correspond to the correlated camera groups used in the simulations, with repeated cameras reflecting overlapping group definitions.}
    \label{fig:warehouse_camera_scenario}
\end{figure*}

\begin{table}[t]
    \centering
    \caption{Simulation Parameters}
    \label{tab:simulation_parameters}
    \footnotesize
    \setlength{\tabcolsep}{3pt}
    \begin{tabular}{|l|c|}
        \hline
        \textbf{System Parameter} & \textbf{Value} \\ \hline
        Number of Cameras ($M$) & 16 \\ \hline
        Number of Camera Groups & 4 \\ \hline
        Number of Encoder Variants ($V$) & 9 \\ \hline
        Latent Dimension & 256 \\ \hline
        TWI Candidate Set ($\Omega$) & $\{0,1,2,4,8,16,32\}$ \\ \hline
        Rate Threshold ($R_{\mathrm{th}}$) & 0.75 \\ \hline
        Normalized Max. Power ($P_{\max}$) & 1 \\ \hline
        Normalized Bandwidth ($B$) & 1 \\ \hline
        AoL-Resource Weights ($\omega_{\mathrm{AoL}},\omega_{\mathrm{res}}$) & $(0.5, 0.5)$ \\ \hline
        \hline\hline
        \textbf{Offline Latent-Model Parameter} & \textbf{Value} \\ \hline
        JEPA Batch Size & $2$ \\ \hline
        JEPA Training Steps & $500$ \\ \hline
        JEPA Learning Rate & $10^{-3}$ \\ \hline
        Predictor History Length & 1 \\ \hline
        Predictor Training Epochs & $50$ \\ \hline
        Predictor Learning Rate & $5\times10^{-4}$ \\ \hline
        Cantelli Parameter ($\delta_m$) & $0.1$ \\ \hline
        \hline\hline
        \textbf{Lyapunov Parameter} & \textbf{Value} \\ \hline
        Lyapunov Control Weight ($V_{\mathrm{L}}$) & 1.0 \\ \hline
        Base Violation Penalty ($\lambda_0$) & 1.0 \\ \hline
        Queue Update Stepsize ($\eta_Q$) & 1.0 \\ \hline
        \hline\hline
        \textbf{PPO Hyperparameter} & \textbf{Value} \\ \hline
        Actor/Critic Learning Rate & $7.5\times10^{-5}$ / $1.5\times10^{-4}$ \\ \hline
        PPO Clip Parameter & 0.10 \\ \hline
        Target KL & 0.02 \\ \hline
    \end{tabular}
\end{table}

\section{Numerical Results} \label{Sec:Numerical_Results}
\subsection{Simulation Setup}
We consider a warehouse multi-camera wireless perception scenario based on the Warehouse MultiCam RF Dataset. The dataset contains $M=16$ fixed RGB cameras arranged over four viewing directions with four cameras per direction. The cameras observe a mobile robot moving along a predefined trajectory in the warehouse. The RF measurements in the dataset are generated using Sionna for a BS equipped with a $16\times 16$ antenna array and $16$ OFDM subcarriers. We define the correlated camera groups based on viewing-direction similarity as $\mathcal{G}_1=\{\mathrm{cam00}\text{--}\mathrm{cam03},\mathrm{cam15}\}$, $\mathcal{G}_2=\{\mathrm{cam12}\text{--}\mathrm{cam15}\}$, $\mathcal{G}_3=\{\mathrm{cam08}\text{--}\mathrm{cam12}\}$, and $\mathcal{G}_4=\{\mathrm{cam04}\text{--}\mathrm{cam07}\}$. Fig.~\ref{fig:warehouse_camera_scenario} shows synchronized camera views at frame 64, arranged by the correlated groups. The blue arrows indicate the robot moving direction. We adopt the TWI candidate set $\Omega=\{0,1,2,4,8,16,32\}$, which covers progressively longer integration windows without enumerating every intermediate length and keeps the task-level TWI search tractable. 

For visual encoding, we train a camera-agnostic JEPA encoder family offline using synchronized frames extracted from the dataset videos. The family contains $V=9$ variants formed by three encoder capacity levels and three input resolutions. Each encoder maps a camera frame to a $d_z$-dimensional latent representation, and the measured inference latency is used as the camera-side encoder cost $c_v^{\mathrm{enc}}$. Disjoint training, calibration, and test splits are used for the BS-side predictor $q_m(\cdot)$ and the uncertainty calibrator. Table~\ref{tab:simulation_parameters} summarizes the simulation parameters and offline model settings. Table~\ref{tab:encoder_variants} lists the camera-side encoder variants.

We evaluate a multi-camera perception task over $1000$ communication slots. In this use case, each instance is a scene-understanding update for monitoring the moving robot under varying channel conditions and camera availability. Four deterministic camera-outage windows are imposed, as summarized in Table~\ref{tab:outage_configuration}. During each outage window, the unavailable cameras cannot upload fresh latents, which models temporary camera or communication outages. The BS must therefore maintain their latent beliefs using available camera updates and the prediction mechanism. The outage windows consist of two short $30$-slot outages and two long $100$-slot outages, creating alternating normal, outage, and recovery phases for evaluating the runtime behavior of the proposed TWI control. The outage windows and unavailable camera pairs are fixed across all compared methods, so that performance differences reflect how each method maintains latent beliefs and controls resources.

\subsection{Benchmark Configurations}
We compare CoLA with three benchmarks.
\begin{itemize}
    \item \textbf{Un-CoLA} sets $\mathcal{C}_m(t)=\emptyset$, so prediction relies only on $\mathcal{D}_m(t-1)$. This comparison isolates the benefit of information from correlated cameras.

    \item \textbf{CoLA-OMA} retains CoLA's prediction mechanism and adaptive TWI selection, but replaces NOMA with orthogonal multiple access (OMA) and equal-power allocation among scheduled cameras. This comparison quantifies the benefit of NOMA-based multiple access.

    \item \textbf{CoLA-Token-ViT} replaces the compact JEPA latent representation with DINOv2 patch tokens~\cite{ref:dinov2}, which provide discriminative features for downstream perception without reconstructing the input image. This comparison assesses the advantage of the proposed compact latent representation over a token-based representation.
\end{itemize}

\begin{table*}[t]
    \centering
    \caption{Camera-Side JEPA Encoder Variants}
    \label{tab:encoder_variants}
    \footnotesize
    \setlength{\tabcolsep}{3pt}
    \begin{tabular}{|c|c|c|c|c|c|c|c|c|c|}
        \hline
        \multirow{2}{*}{\textbf{Input Resolution}} 
        & \multicolumn{3}{c|}{\textbf{Small Encoder (ResNet-18)}} 
        & \multicolumn{3}{c|}{\textbf{Base Encoder (ResNet-34)}} 
        & \multicolumn{3}{c|}{\textbf{Large Encoder (ResNet-50)}} \\ \cline{2-10}
        & \textbf{Variant} & \textbf{Inference Latency} & \textbf{Model Size}
        & \textbf{Variant} & \textbf{Inference Latency} & \textbf{Model Size}
        & \textbf{Variant} & \textbf{Inference Latency} & \textbf{Model Size} \\ \hline
        $128\times128$ 
        & 0 & 2.837 ms & 44.471 MB
        & 3 & 4.814 ms & 83.091 MB
        & 6 & 6.189 ms & 93.243 MB \\ \hline
        $224\times224$ 
        & 1 & 3.280 ms & 44.472 MB
        & 4 & 4.431 ms & 83.091 MB
        & 7 & 6.179 ms & 93.243 MB \\ \hline
        $384\times384$ 
        & 2 & 3.487 ms & 44.471 MB
        & 5 & 5.333 ms & 83.091 MB
        & 8 & 7.414 ms & 93.242 MB \\ \hline
    \end{tabular}
\end{table*}

\begin{table}[t]
    \centering
    \caption{Camera-Outage Configuration}
    \label{tab:outage_configuration}
    \footnotesize
    \setlength{\tabcolsep}{3pt}
    \begin{tabular}{|l|c|c|c|}
        \hline
        \textbf{Window} & \textbf{Unavailable Cameras} & \textbf{Slot Interval} & \textbf{Duration} \\ \hline
        Outage 1 & $\{\mathrm{cam00},\mathrm{cam04}\}$ & $[120,150)$ & $30$ slots \\ \hline
        Outage 2 & $\{\mathrm{cam09},\mathrm{cam12}\}$ & $[190,290)$ & $100$ slots \\ \hline
        Outage 3 & $\{\mathrm{cam02},\mathrm{cam06}\}$ & $[500,530)$ & $30$ slots \\ \hline
        Outage 4 & $\{\mathrm{cam11},\mathrm{cam15}\}$ & $[750,850)$ & $100$ slots \\ \hline
    \end{tabular}
\end{table}

\begin{figure}[!t]
    \centering
    \includegraphics[width=2.8in]{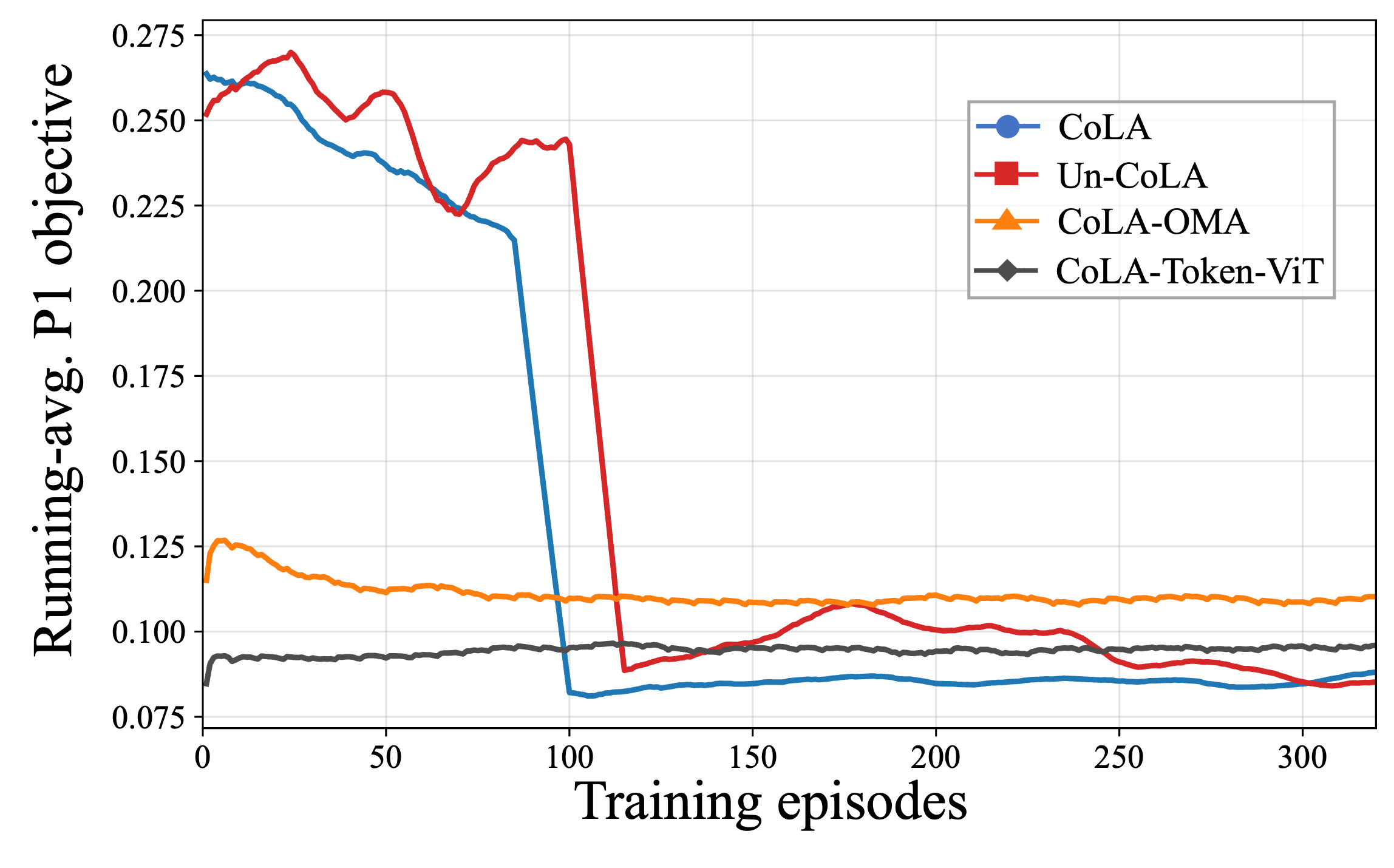}
    \caption{Training convergence of the slot-level PPO resource-control policy under the heuristic TWI selection policy.}
    \label{fig:training_convergence}
\end{figure}

\begin{figure}[!t]
    \centering
    \subfloat[Running-average $\mathbf{P1}$ objective.]{
        \includegraphics[width=3in]{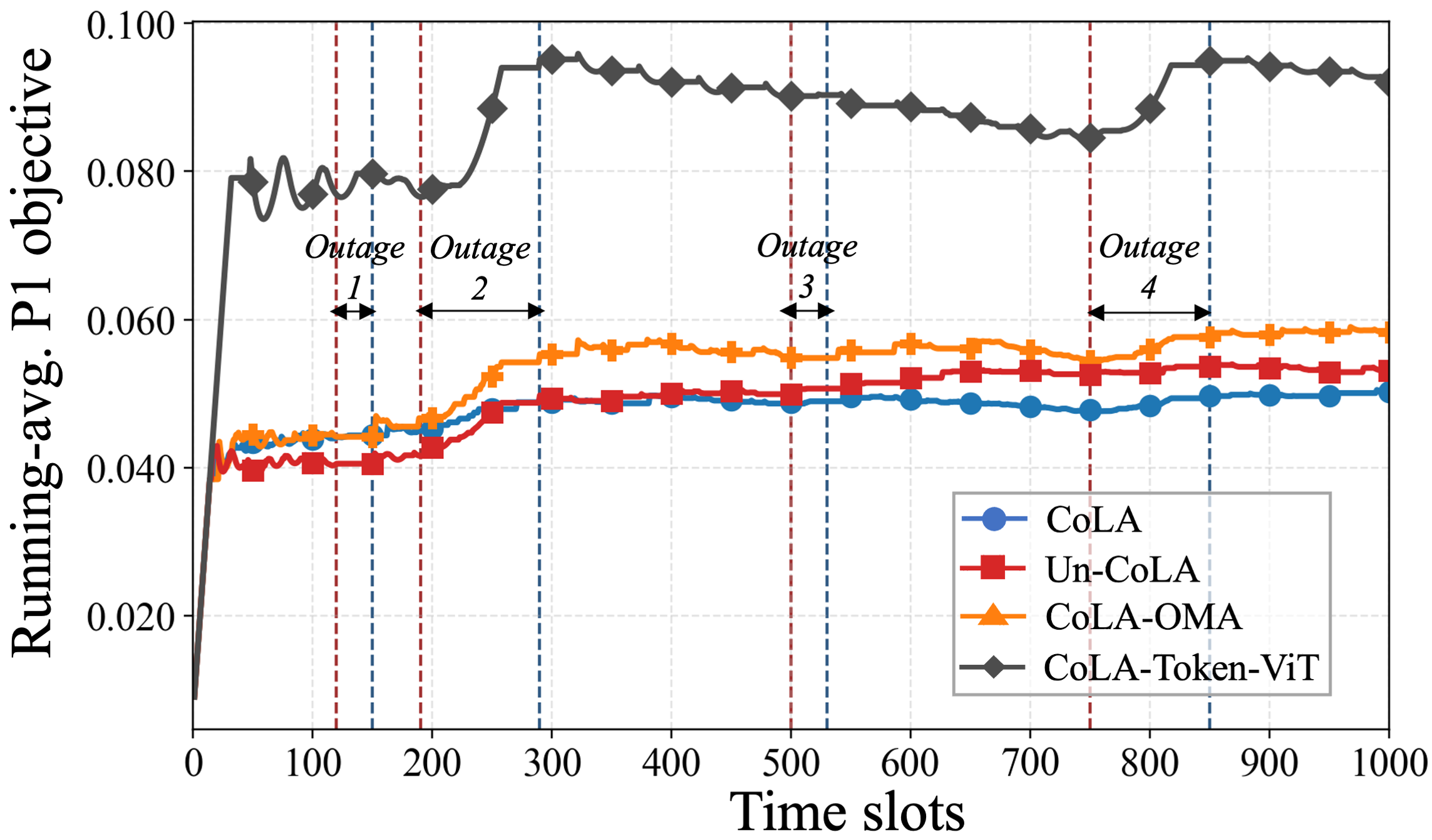}
        \label{fig:performance_reliability_objective}
    }
    
    \subfloat[CCDF of $u_m(T_k)/u_m^{\mathrm{tar}}$.]{
        \includegraphics[width=3in]{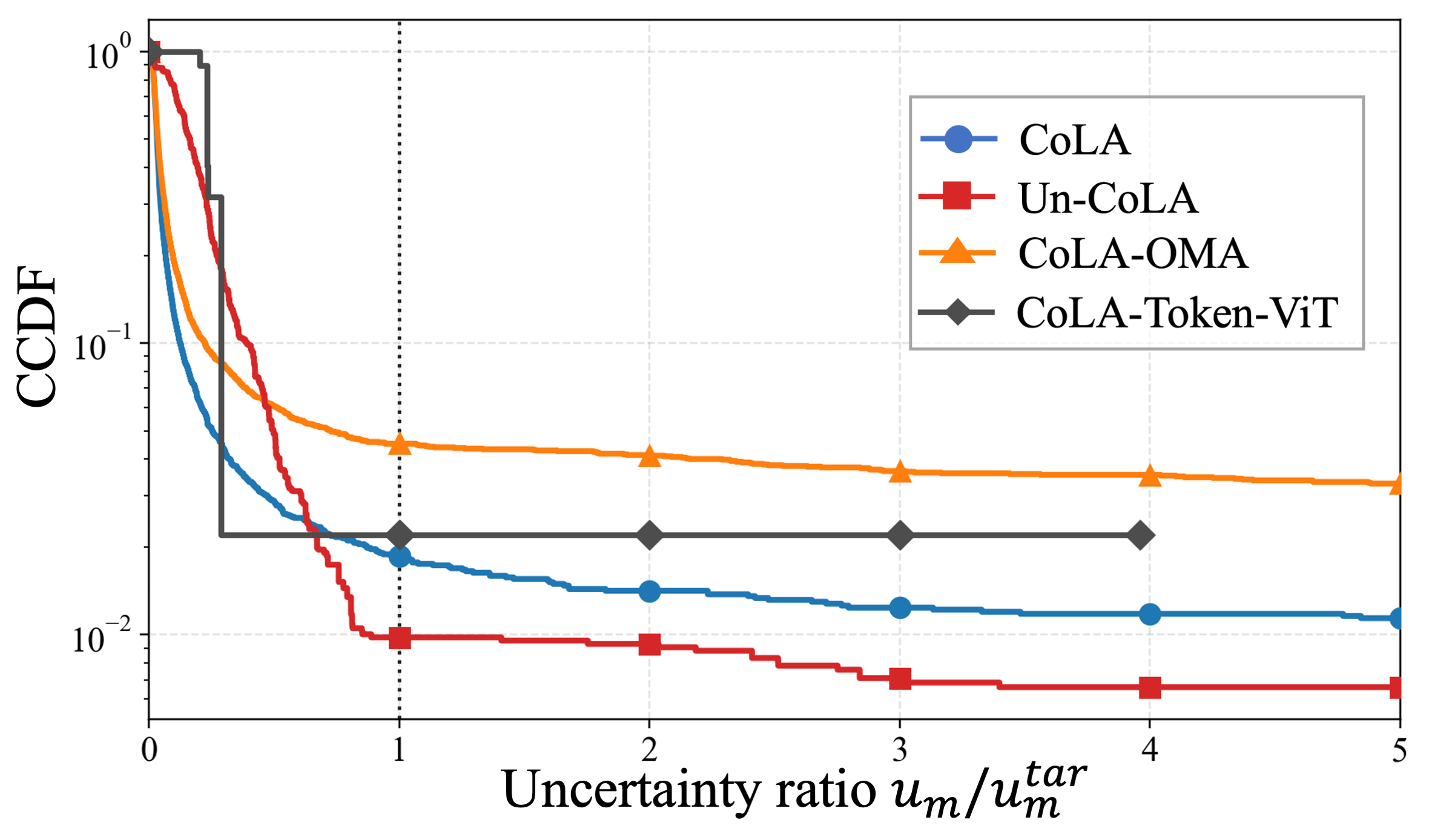}
        \label{fig:performance_reliability_cdf}
    }
    \caption{Performance and reliability comparison under the camera-outage setting. (a) Running-average $\mathbf{P1}$ objective over time slots. (b) CCDF of the uncertainty ratio $u_m(T_k)/u_m^{\mathrm{tar}}$ evaluated at the task-commitment slot $T_k$ on a logarithmic vertical scale. The vertical dotted line marks the uncertainty-violation threshold $u_m/u_m^{\mathrm{tar}}=1$.}
    \label{fig:performance_reliability}
\end{figure}

\subsection{Training Convergence}
We first examine the training convergence of the slot-level resource-control policy. During PPO training, the TWI is selected by a heuristic policy that cycles the nonzero TWI candidates in $\Omega$. Fig.~\ref{fig:training_convergence} shows the PPO training objective for CoLA and the benchmarks. The training objectives reach stable levels for all schemes, with CoLA and Un-CoLA settling at lower levels than CoLA-OMA and CoLA-Token-ViT. The trained slot-level policies are then used for Lyapunov-based TWI selection in the following subsections.

\begin{figure}[!t]
    \centering
    \subfloat[Short camera-outage windows. \label{fig:short_outage_twi}]{
        \includegraphics[width=3.5in]{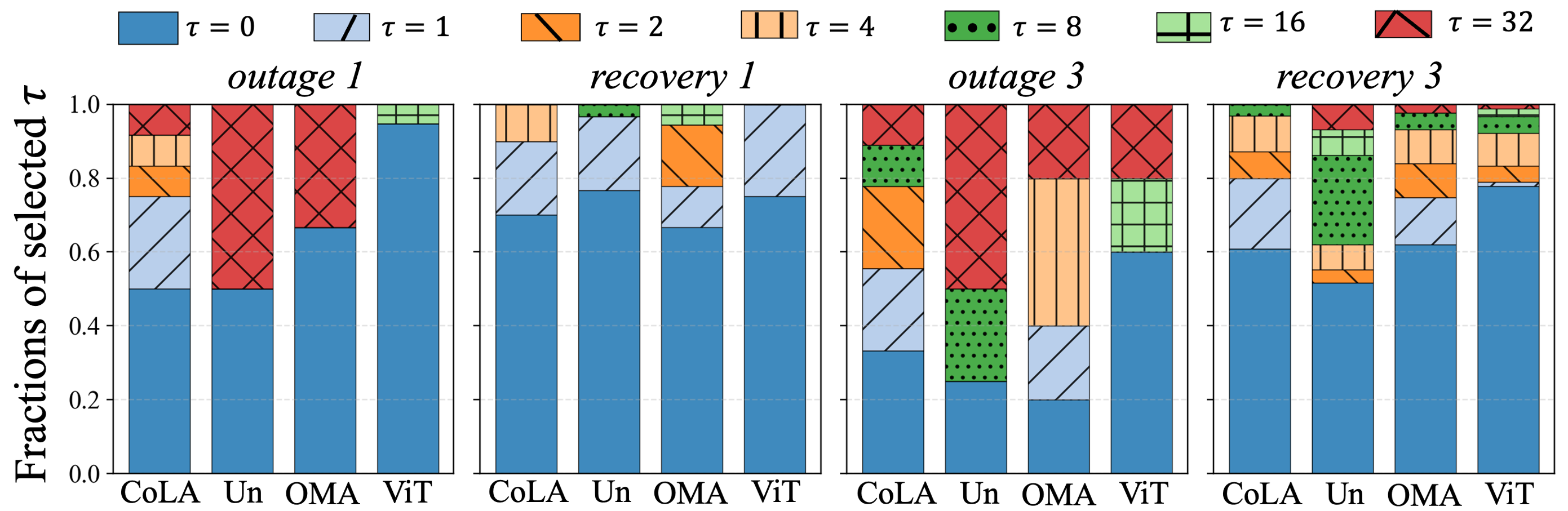}}
        
    \subfloat[Long camera-outage windows. \label{fig:long_outage_twi}]{
        \includegraphics[width=3.5in]{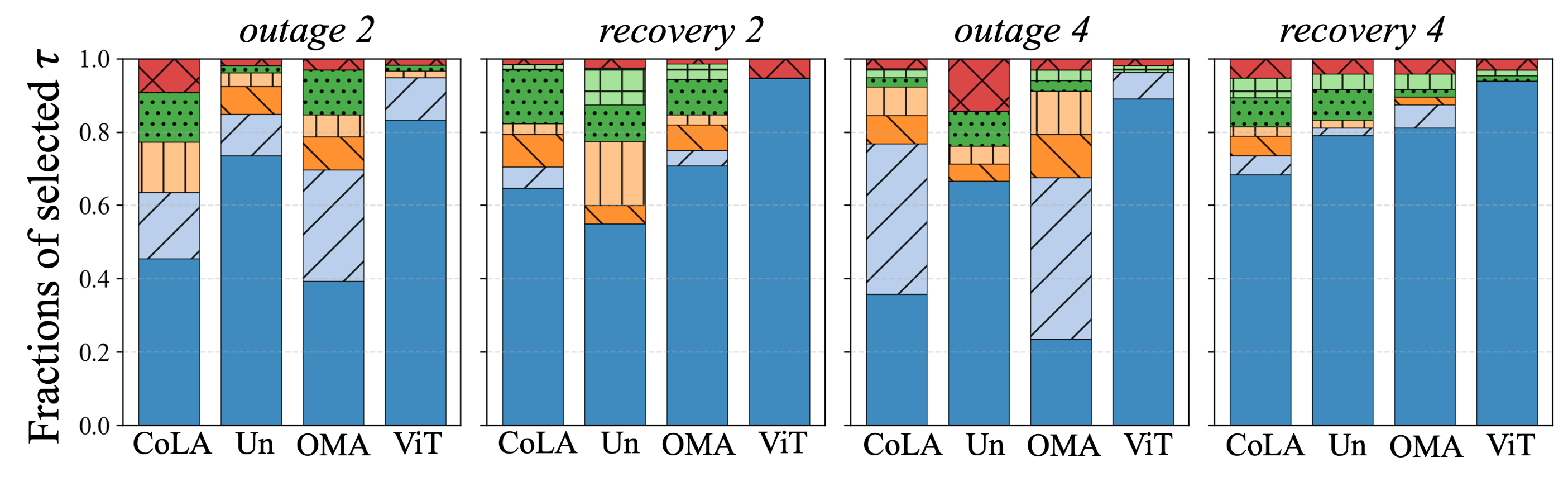}}

    \caption{TWI selection ratios under short and long camera-outage windows. The labels Un, OMA, and ViT denote Un-CoLA, CoLA-OMA, and CoLA-Token-ViT, respectively.}
    \label{fig:outage_twi}
\end{figure}

\subsection{Performance and Reliability Tradeoff}
We compare CoLA and the three benchmarks under the camera-outage setting. Fig.~\ref{fig:performance_reliability}\subref{fig:performance_reliability_objective} shows the running-average $\mathbf{P1}$ objective, where a lower value indicates a better joint AoL-resource tradeoff. CoLA achieves a lower running-average $\mathbf{P1}$ objective than Un-CoLA, CoLA-OMA, and CoLA-Token-ViT, with reductions of $2.7\%$, $10.5\%$, and $45.1\%$, respectively. The gap relative to CoLA-OMA shows the benefit of NOMA-based concurrent latent transmissions under limited communication opportunities, while the gap relative to CoLA-Token-ViT indicates that the compact JEPA latent representation is more effective than the token-based representation for maintaining BS-side latent beliefs under limited resources. 

Fig.~\ref{fig:performance_reliability}\subref{fig:performance_reliability_cdf} shows the complementary cumulative distribution function (CCDF) of the uncertainty ratio $u_m(T_k)/u_m^{\mathrm{tar}}$ evaluated at the task-commitment slot $T_k$ on a logarithmic vertical scale. At the threshold $u_m/u_m^{\mathrm{tar}}=1$ in (\ref{P1b}), the CCDF equals the uncertainty-violation rate. CoLA yields a violation rate of $1.87\%$, lower than those of CoLA-OMA and CoLA-Token-ViT. Un-CoLA attains the lowest violation rate, but incurs higher resource consumption than CoLA. Compared with Un-CoLA, CoLA reduces the mean resource cost by $26.5\%$ while maintaining a lower running-average $\mathbf{P1}$ objective. The CoLA-Token-ViT curve terminates at its largest observed ratio because its empirical CCDF is zero beyond that point and cannot be represented on the logarithmic scale.

\begin{figure}[!tbp]
    \centering
    \subfloat[Short camera-outage windows. \label{fig:short_outage_aol}]{
        \includegraphics[width=\columnwidth]{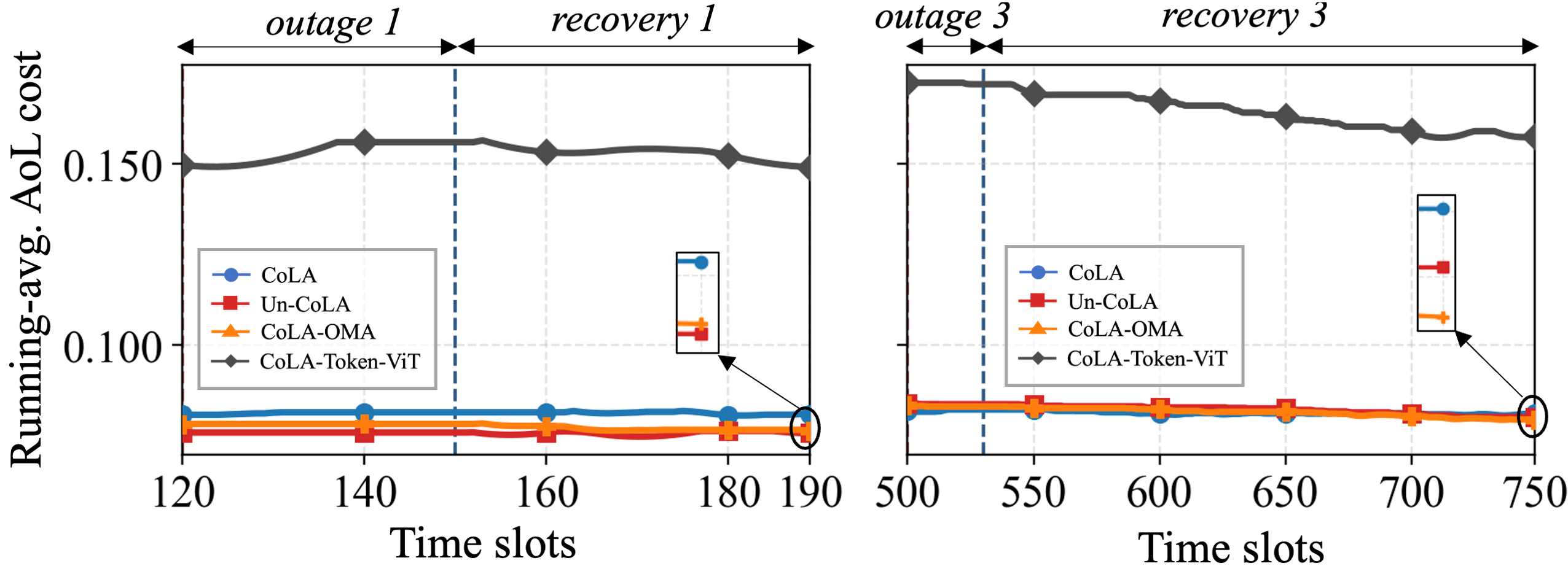}}

    \vspace{0.4em}

    \subfloat[Long camera-outage windows. \label{fig:long_outage_aol}]{
        \includegraphics[width=\columnwidth]{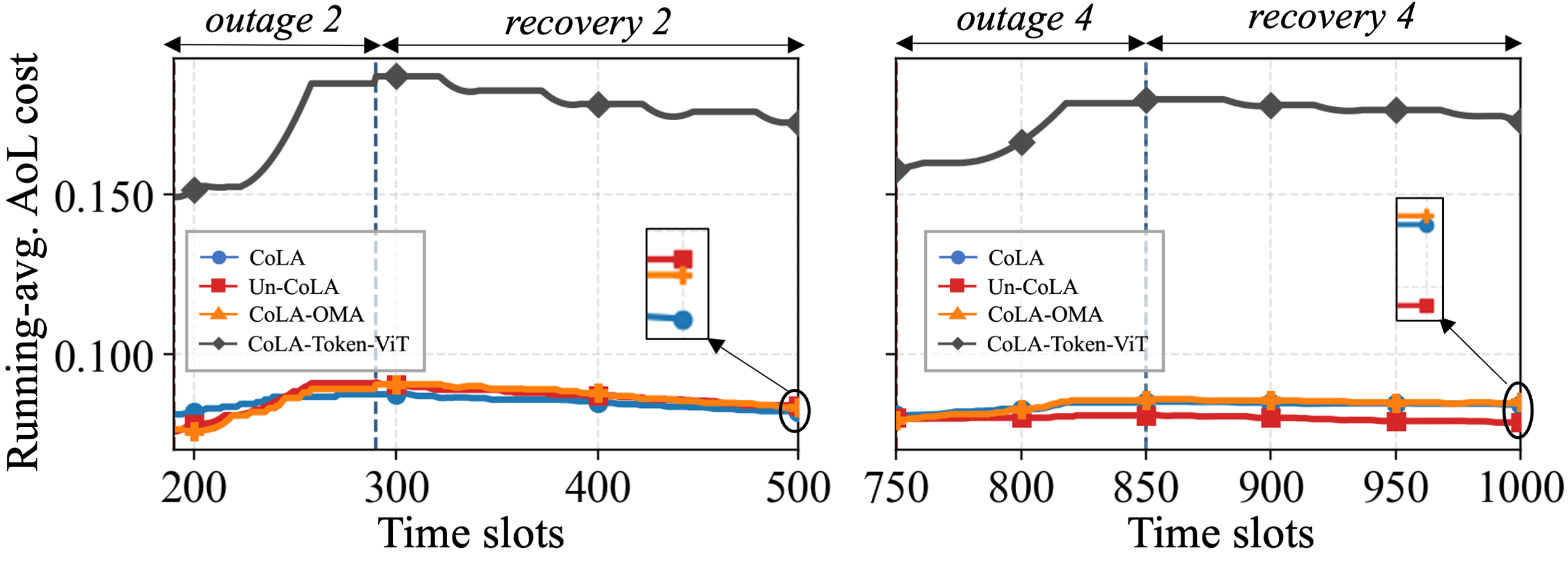}}

    \caption{Running-average AoL under short and long camera-outage windows.}
    \label{fig:outage_aol}
\end{figure}

\begin{figure}[!tbp]
    \centering
    \subfloat[Short camera-outage windows. \label{fig:short_outage_resource}]{
        \includegraphics[width=\columnwidth]{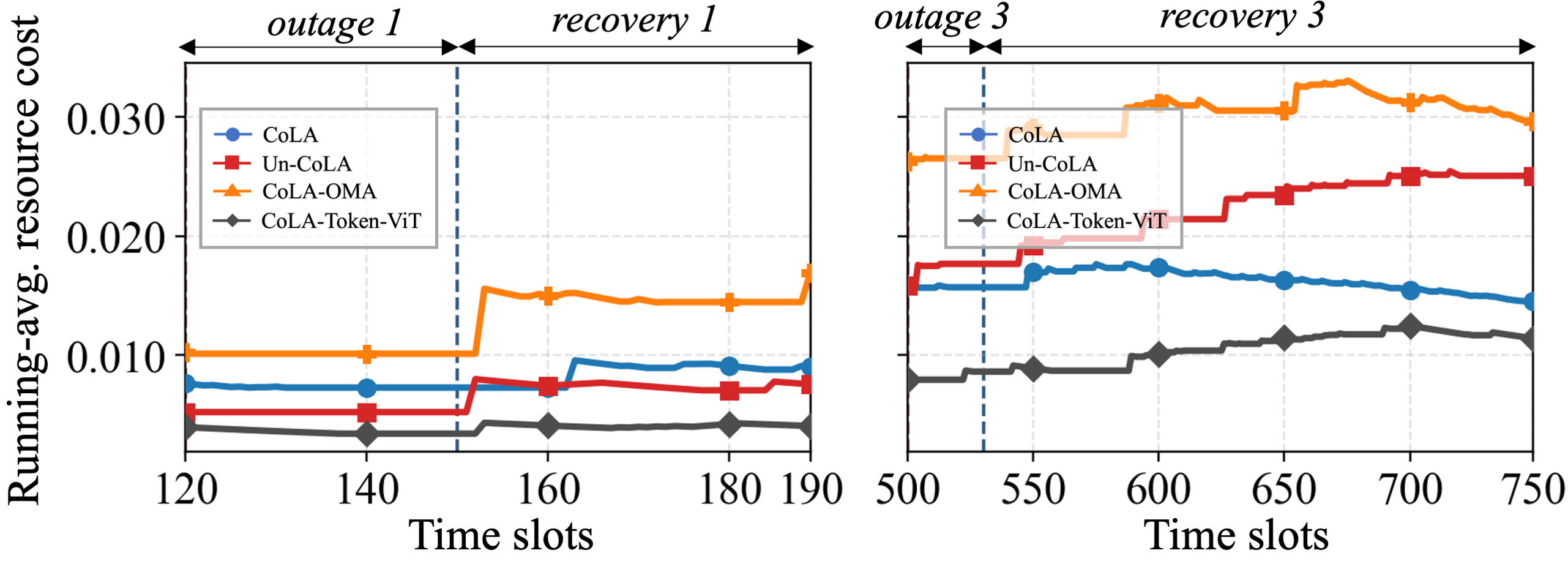}}

    \vspace{0.4em}

    \subfloat[Long camera-outage windows. \label{fig:long_outage_resource}]{
        \includegraphics[width=\columnwidth]{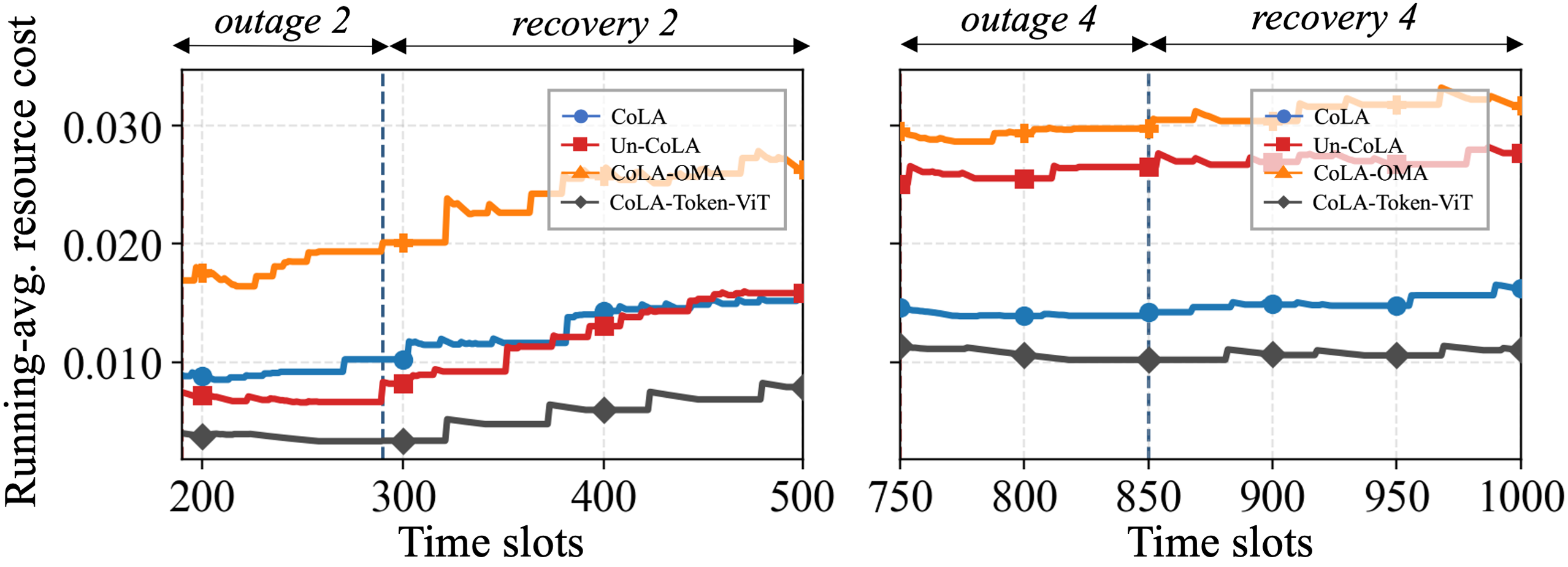}}

    \caption{Running-average resource cost under short and long camera-outage windows.}
    \label{fig:outage_resource}
\end{figure}

\subsection{Adaptive TWI Selection Under Camera Outages}
\label{subsec:adaptive_twi_outage}
We next examine how CoLA and the benchmarks adapt the TWI during camera-outage and recovery stages. For each outage window, we refer to the interval after the unavailable cameras become available again as the \textit{recovery stage}.

Fig.~\ref{fig:outage_twi} compares the TWI selection ratios in the outage and recovery stages. For CoLA, $\tau=0$ is selected more frequently in every recovery stage than in the corresponding outage stage, whereas positive TWIs are selected more frequently during outages. The CoLA distributions also differ between outages 2 and 4, showing that the selected TWI is not determined by outage duration alone. Compared with CoLA, Un-CoLA and CoLA-OMA place substantially more probability on $\tau=32$ during the short outages, while CoLA-Token-ViT remains more concentrated on $\tau=0$ across most stages. These selection patterns complement the preceding objective results: CoLA adapts the integration horizon across stages while retaining the lowest overall $\mathbf{P1}$ objective among the compared schemes.

\begin{table}[!t]
    \centering
    \caption{Resource-Cost Decomposition Under Camera-Outage and Recovery Stages}
    \label{tab:resource_decomposition_stage}
    \scriptsize
    \setlength{\tabcolsep}{2.2pt}
    \begin{tabular}{llccc}
        \toprule
        Method 
        & Stage 
        & Selected enc. 
        & Mean enc. cost 
        & Mean power cost \\
        \midrule
        \multirow{4}{*}{CoLA}
        & Short outage 
        & 1, 4 
        & 0.037 
        & 0.336 \\
        & Short recovery 
        & 1, 4, 6 
        & 0.061 
        & 0.619 \\
        & Long outage 
        & 1, 4 
        & 0.038 
        & 0.266 \\
        & Long recovery 
        & 1, 4, 6 
        & 0.064 
        & 0.641 \\
        \midrule
        \multirow{4}{*}{CoLA-OMA}
        & Short outage 
        & 2, 3 
        & 0.060 
        & 0.983 \\
        & Short recovery 
        & 2, 3 
        & 0.134 
        & 1.000 \\
        & Long outage 
        & 2, 3 
        & 0.031 
        & 0.782 \\
        & Long recovery 
        & 2, 3 
        & 0.126 
        & 1.000 \\
        \midrule
        \multirow{4}{*}{Un-CoLA}
        & Short outage 
        & 4 
        & 0.069 
        & 0.517 \\
        & Short recovery 
        & 4, 5 
        & 0.091 
        & 0.741 \\
        & Long outage 
        & 4, 5 
        & 0.071 
        & 0.521 \\
        & Long recovery 
        & 4, 5 
        & 0.090 
        & 0.740 \\
        \midrule
        \multirow{4}{*}{\makecell[l]{CoLA-\\Token-ViT}}
        & Short outage 
        & 5 
        & 0.020 
        & 0.315 \\
        & Short recovery 
        & 5 
        & 0.048 
        & 0.751 \\
        & Long outage 
        & 5 
        & 0.007 
        & 0.101 \\
        & Long recovery 
        & 5, 6 
        & 0.051 
        & 0.747 \\
        \bottomrule
    \end{tabular}
\end{table}

\subsection{Outage-Stage Component Analysis}
\label{subsec:outage_stage_component_analysis}
We further decompose the camera-outage behavior by stage and compare the running-average AoL and resource costs in Figs.~\ref{fig:outage_aol} and \ref{fig:outage_resource} to identify the performance gain of CoLA.

For the short outage windows, Fig.~\ref{fig:outage_aol}\subref{fig:short_outage_aol} shows that the JEPA-based schemes achieve similar running-average AoL costs, whereas CoLA-Token-ViT stays at a much higher AoL cost despite its lower resource cost in Fig.~\ref{fig:outage_resource}\subref{fig:short_outage_resource}. The gap indicates that the token-based representation is less effective in maintaining low AoL under limited resources during camera outages. Among these schemes, CoLA achieves an AoL cost similar to that of CoLA-OMA with lower resource cost, while CoLA differs from Un-CoLA mainly in resource cost after Outage 3. Therefore, the short outage windows primarily reveal resource-cost differences.

For the long outage windows in Fig.~\ref{fig:outage_aol}\subref{fig:long_outage_aol}, CoLA keeps its running-average AoL cost close to the corresponding costs of Un-CoLA and CoLA-OMA, whereas CoLA-Token-ViT remains at a much higher AoL cost. The sustained high AoL cost of CoLA-Token-ViT suggests that the token-based representation is less effective in maintaining low AoL during extended camera outages. Fig.~\ref{fig:outage_resource}\subref{fig:long_outage_resource} further shows that, among the JEPA-based schemes, CoLA incurs a lower resource cost than CoLA-OMA and Un-CoLA while maintaining a comparable AoL cost.

Table~\ref{tab:resource_decomposition_stage} further decomposes the resource cost into encoder and power components. Compared with CoLA-OMA, CoLA mainly reduces the power cost, supporting the role of NOMA-based multiple access in lowering the power-side burden. Compared with Un-CoLA, CoLA reduces both encoder and power costs, indicating that information from correlated cameras reduces the resources required to sustain comparable AoL. CoLA-Token-ViT has low encoder and outage-stage power costs, but its high AoL cost in Fig.~\ref{fig:outage_aol} indicates that low resource consumption alone is insufficient without reliable BS-side latent beliefs.

\begin{figure}[!t]
    \centering
    \subfloat[Running-average $\mathbf{P1}$ objective.]{
        \includegraphics[width=3in]{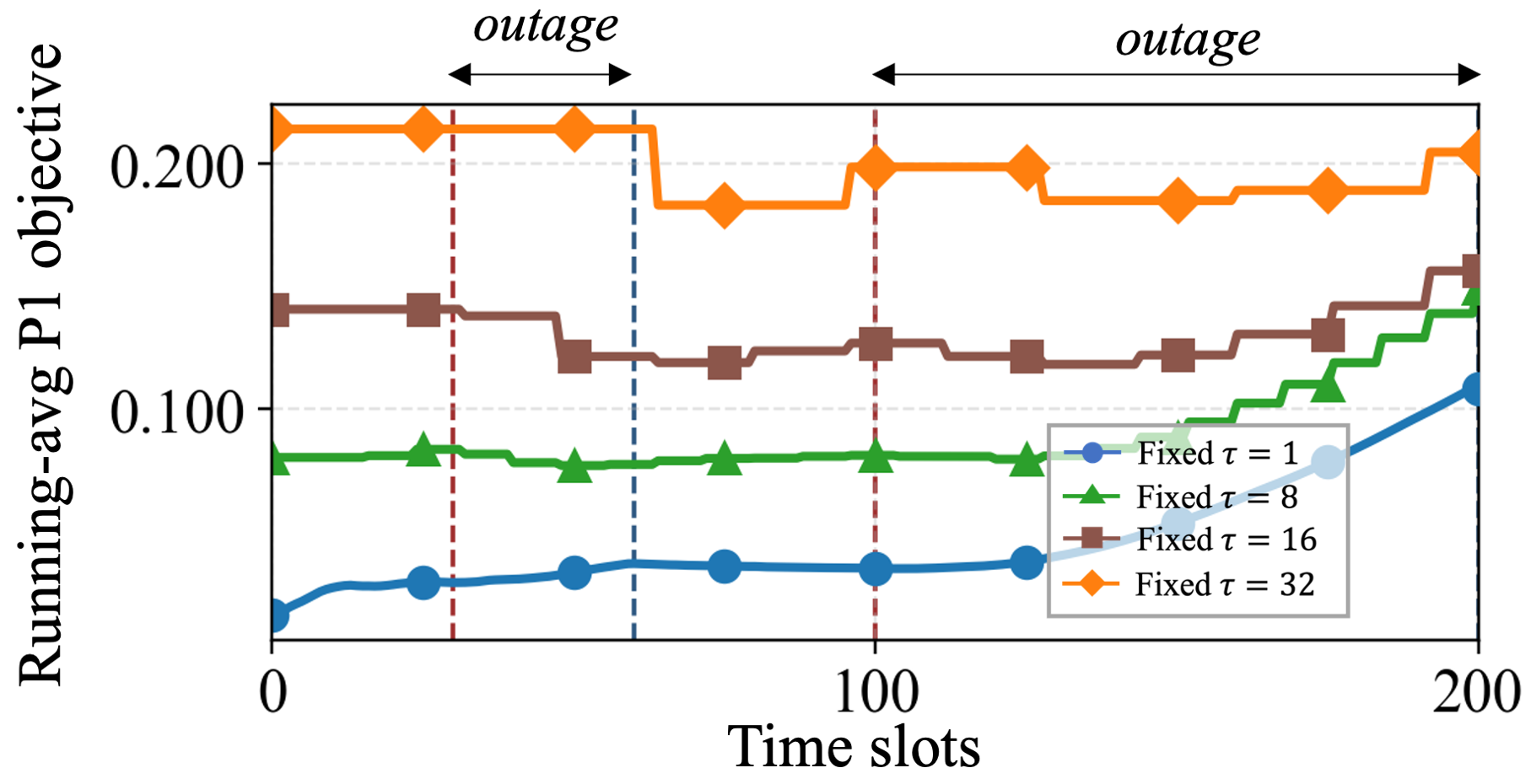}
        \label{fig:twi_selection_objective}
    }
    
    \subfloat[Running-average AoL.]{
        \includegraphics[width=3in]{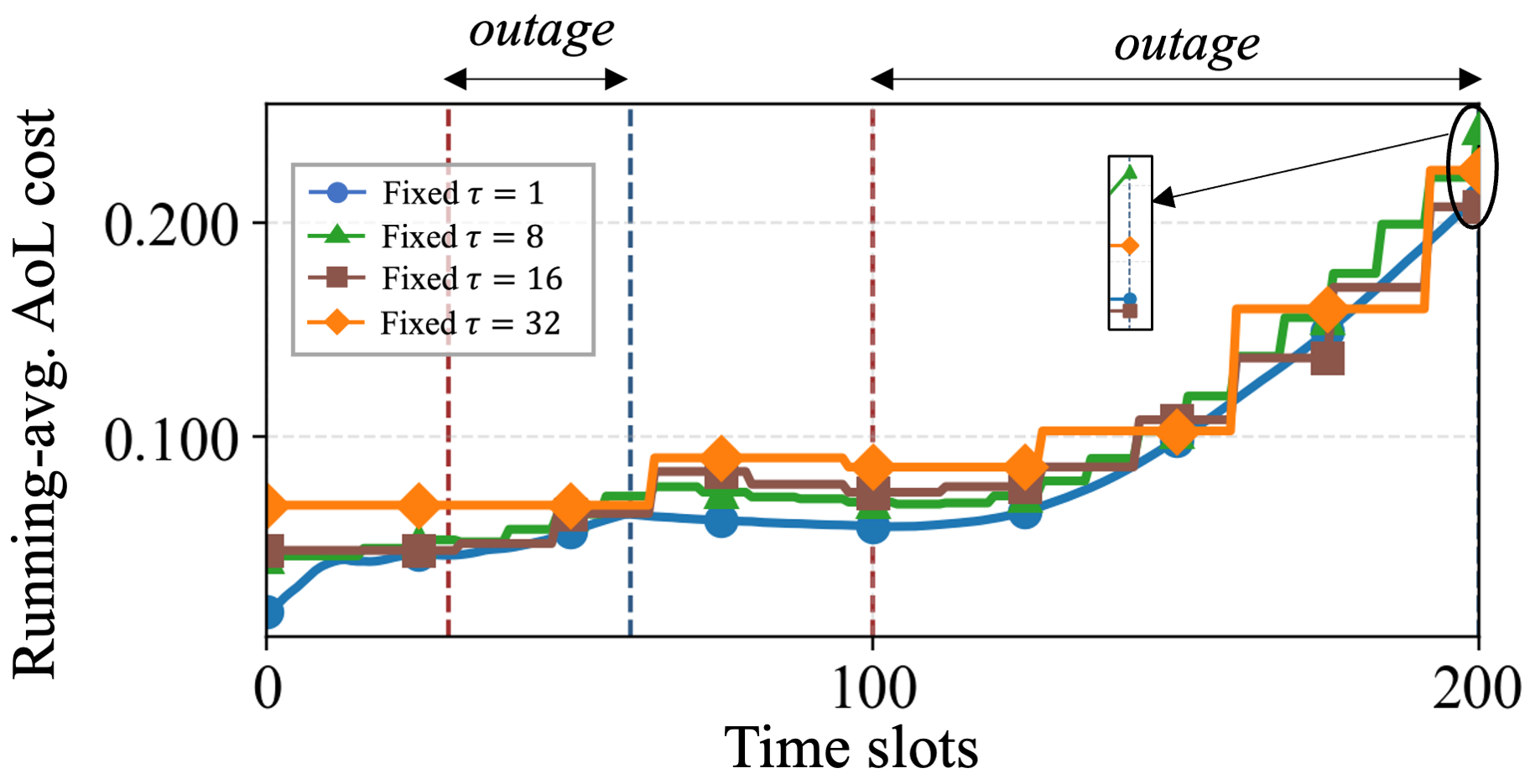}
        \label{fig:twi_selection_aol}
    }
    
    \subfloat[CCDF of $u_m(T_k)/u_m^{\mathrm{tar}}$.]{
        \includegraphics[width=3in]{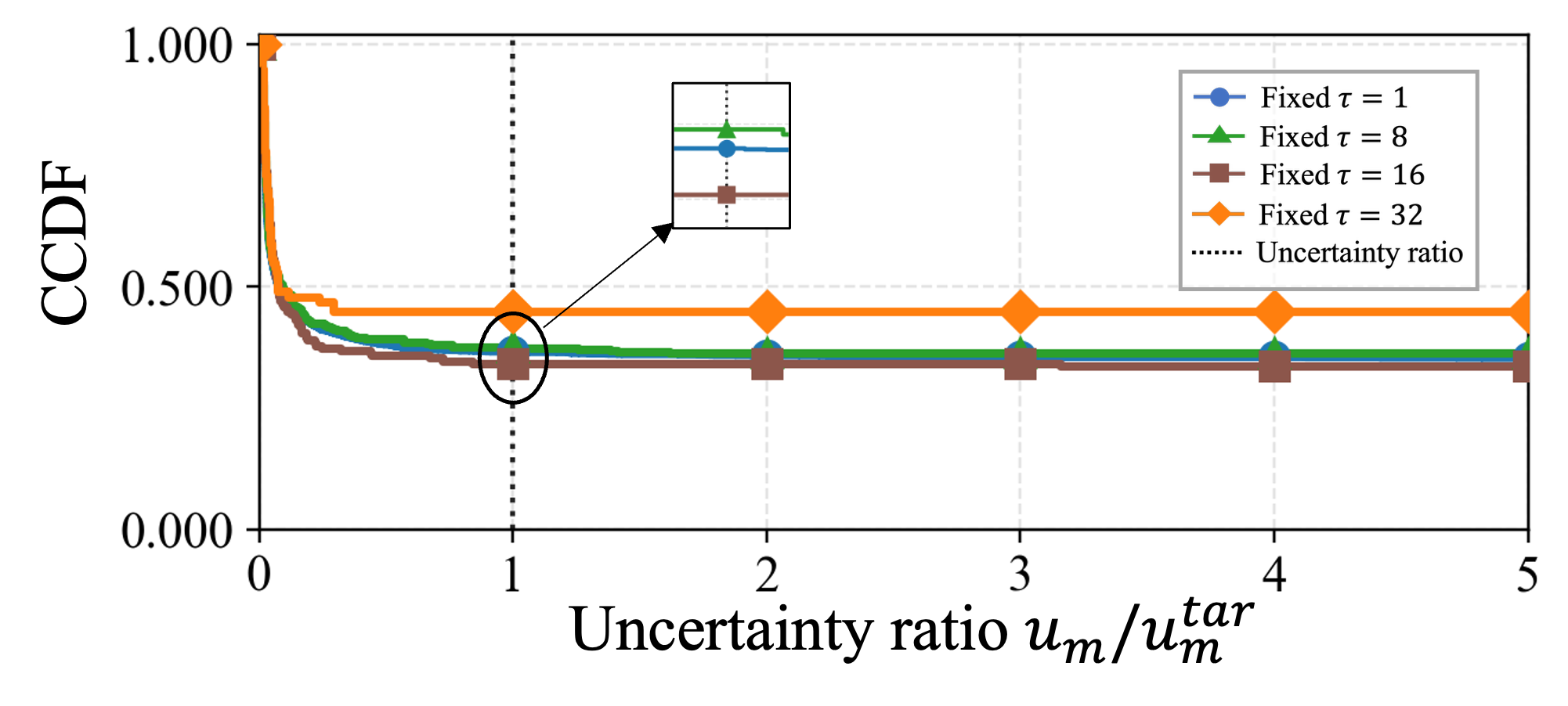}
        \label{fig:twi_selection_uncertainty}
    }
    \caption{Different fixed TWI policies with short and long camera-outage stages. (a) Running-average $\mathbf{P1}$ objective. (b) Running-average AoL. (c) CCDF of the uncertainty ratio $u_m(T_k)/u_m^{\mathrm{tar}}$ evaluated at the task-commitment slot $T_k$.}
    \label{fig:twi_selection_policy}
\end{figure}

\subsection{Impact of TWI Selection Policy}
\label{subsec:twi_selection_policy}
To isolate the impact of TWI selection, we evaluate four fixed-TWI policies, each using a constant TWI length $\tau \in \{1,8,16,32\}$. The evaluation uses a $200$-slot setting with a 30-slot short outage and a 100-slot long outage.

Fig.~\ref{fig:twi_selection_policy}\subref{fig:twi_selection_objective} shows that fixed $\tau=1$ achieves the lowest running-average $\mathbf{P1}$ objective among the fixed-TWI policies, whereas larger fixed TWIs lead to higher objective values. Since $\mathbf{P1}$ jointly captures AoL and resource cost, this trend shows that longer fixed TWIs can increase the accumulated cost despite providing more transmission opportunities. Fig.~\ref{fig:twi_selection_policy}\subref{fig:twi_selection_aol} further shows that fixed $\tau=1$ maintains a lower running-average AoL over the evaluated horizon. All fixed-TWI policies nevertheless exhibit AoL growth during the long outage stage because the BS cannot refresh the corresponding latent beliefs through fresh transmissions from unavailable cameras.

Fig.~\ref{fig:twi_selection_policy}\subref{fig:twi_selection_uncertainty} shows that the uncertainty-violation rate does not decrease monotonically as the fixed TWI length increases. Fixed $\tau=16$ reduces the violation rate relative to the fixed $\tau=1$ and $\tau=8$ policies, but incurs a higher AoL-resource objective value. Fixed $\tau=32$ yields the highest violation rate and the highest objective value among the fixed-TWI policies. Thus, no fixed TWI length simultaneously achieves the lowest AoL-resource objective and uncertainty-violation rate. Together with the preceding CoLA results, this finding supports adapting the TWI to latent-belief uncertainty and camera availability.

\begin{figure}[!t]
    \centering
    \subfloat[Running-average $\mathbf{P1}$ objective.]{
        \includegraphics[width=3in]{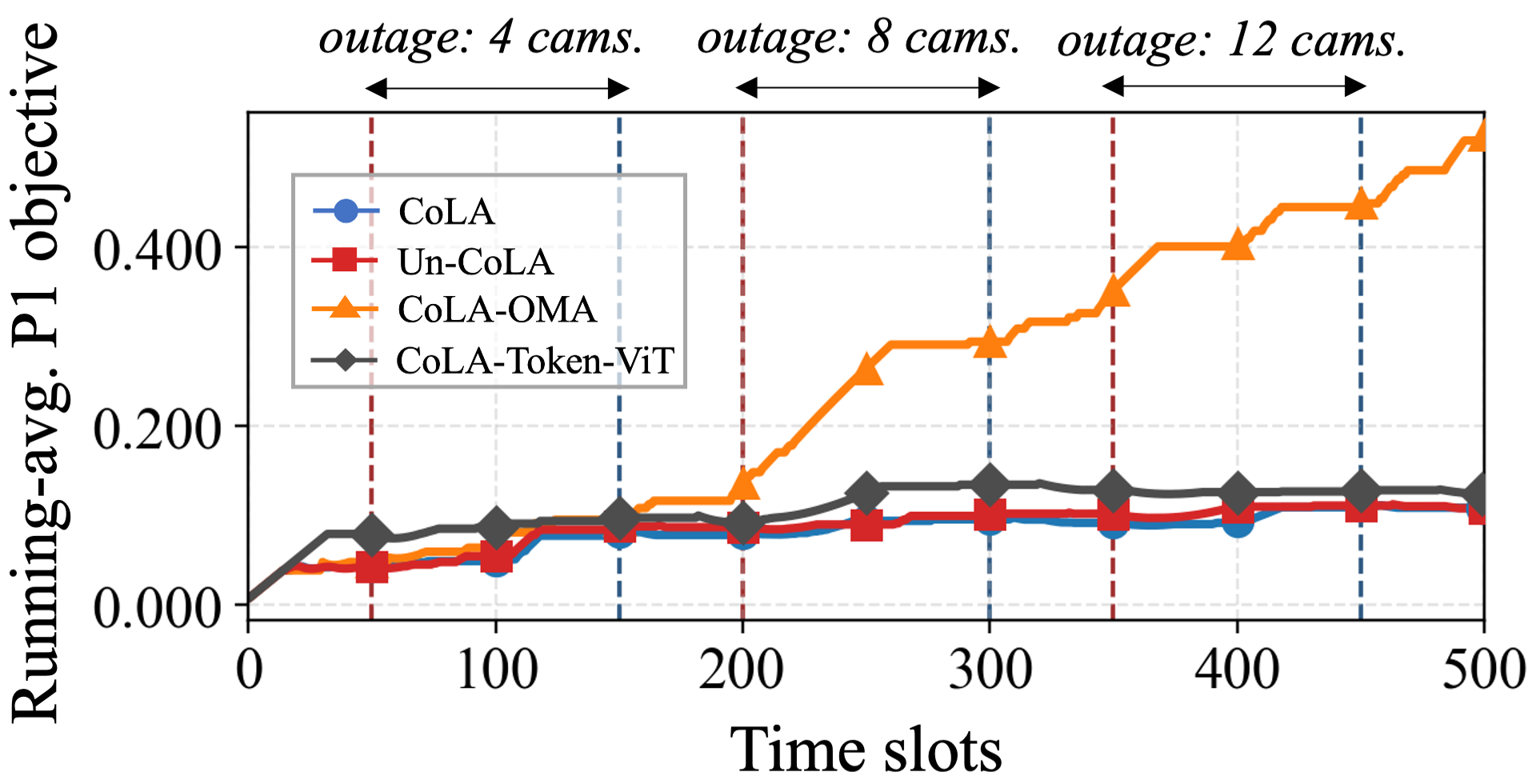}
        \label{fig:outage_severity_objective}
    }
    
    \subfloat[AoL--uncertainty operating points.]{
        \includegraphics[width=3in]{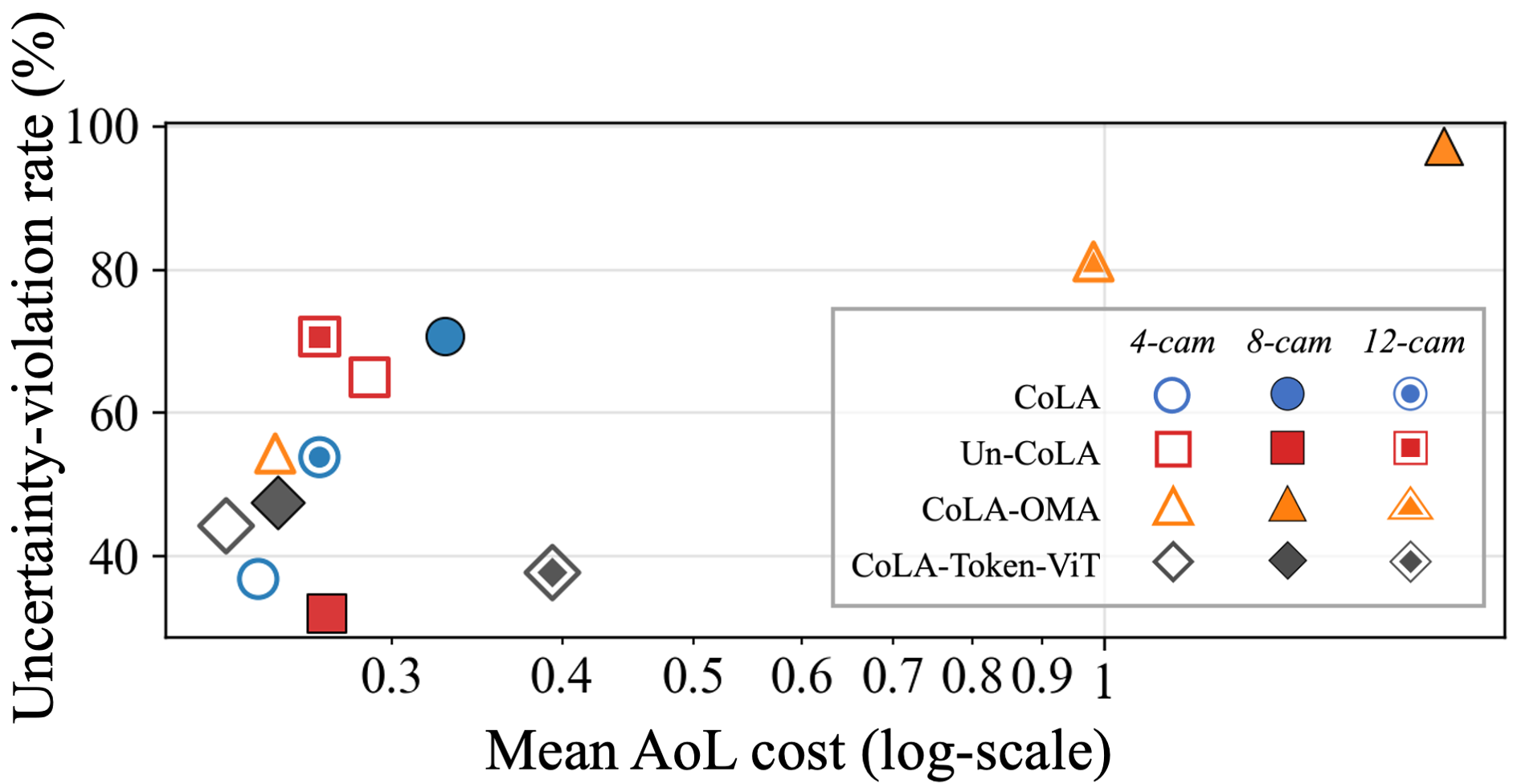}
    \label{fig:outage_severity_aol_uncertainty}
    }
    \caption{Impact of camera-outage severity. (a) Running-average $\mathbf{P1}$ objective. (b) Stage-wise AoL--uncertainty tradeoff.}
    \label{fig:outage_severity}
\end{figure}

\subsection{Impact of Outage Severity}  \label{subsec:outage_severity}
To assess robustness to camera-outage severity, we consider an ablation scenario with alternating outage and recovery stages. The outage stages make $4$, $8$, and $12$ cameras unavailable, corresponding to mild, moderate, and severe outage levels, respectively. Fig.~\ref{fig:outage_severity}\subref{fig:outage_severity_objective} shows the running-average $\mathbf{P1}$ objective under different outage-severity levels. The objective gap widens as the outage severity increases, indicating that maintaining low AoL becomes more costly when more cameras are unavailable. CoLA remains stable across the severity levels, whereas CoLA-OMA exhibits a sharp increase in the objective during the severe outage stage.

Fig.~\ref{fig:outage_severity}\subref{fig:outage_severity_aol_uncertainty} further shows the outage-stage AoL--uncertainty tradeoff. CoLA-OMA moves toward the high-AoL and high-violation region as the outage severity increases, which is consistent with the objective increase in Fig.~\ref{fig:outage_severity}\subref{fig:outage_severity_objective}. In contrast, Un-CoLA and CoLA-Token-ViT show comparable AoL and uncertainty levels, but do not reduce the running-average $\mathbf{P1}$ objective relative to CoLA. The lower $\mathbf{P1}$ objective achieved by CoLA therefore reflects a better balance among AoL, reliability, and resource cost under severe camera outages.

\section{Conclusion} \label{Sec:Conclusion}
We studied adaptive TWI selection for joint AoL-resource minimization in multi-camera wireless perception under limited uplink resources. The formulation couples task-level TWI selection with BS-side latent-belief maintenance and slot-level camera scheduling, encoder selection, and NOMA power allocation under prediction-reliability constraints. CoLA uses Lyapunov optimization for TWI selection based on AoL-resource cost and prediction uncertainty, and PPO for resource control within the selected TWI. Numerical results show that CoLA achieves the lowest running-average $\mathbf{P1}$ objective among the benchmarks while maintaining prediction reliability, with TWI choices that vary across outage and recovery stages rather than following outage duration alone. CoLA also remains stable as camera-outage severity increases. The fixed-TWI results further show that the lowest AoL-resource objective and uncertainty-violation rate are achieved by different window lengths, supporting adaptive TWI selection across operating conditions.


\begin{thebibliography}{99}
    \bibitem{ref:goal_oriented_6g_goals}
    E.C. Strinati, P.D. Lorenzo, V. Sciancalepore, A. Aijaz, M. Kountouris, D. Gündüz, P. Popovski, M. Sana, P.A. Stavrou, B. Soret, N. Cordeschi, S. Scardapane, M. Merluzzi, L. Zanzi, M.B. Renato, T. Quek, N.D. Pietro, O. Forceville, F. Costanzo, and P. Li, "Goal-Oriented and Semantic Communication in 6G AI-Native Networks: The 6G-GOALS Approach," \textit{in Proc. EuCNC/6G Summit}, pp. 1-6, Jun. 2024.

    \bibitem{ref:freshness_to_semantics}
    J. Luo, E. Delfani, M. Salimnejad, and N. Pappas, "From Information Freshness to Semantics of Information and Goal-oriented Communications," \textit{arXiv:2512.12758}. [Online]. Available: https://arxiv.org/abs/2512.12758

    \bibitem{ref:semcom_standardization_2025}
    P. Zhang, X. Xu, M. Sun, H. Gao, N. Ma, X. Wang, R. Zhang, J. Wang, and D. Niyato, "Toward Native AI in 6G Standardization: The Roadmap of Semantic Communication," \textit{IEEE Commun. Stand. Mag.}, pp. 1-11, Mar. 2026. \textit{(Early Access)}

    \bibitem{ref:nokia_physical_ai_ran_2026}
    H. Viswanathan, "Physical AI: Redefining RAN and telco monetization," Nokia, Apr. 8, 2026. [Online]. Available: \url{https://www.nokia.com/blog/physical-ai-redefining-ran-and-telco-monetization/}

    \bibitem{ref:semantic_collaborative_perception_2025}
    Y. Liu, Q. Huang, R. Li, Z. Zhao, S. Zhao, Y. Liu, Y. Zhu, and H. Zhang, "Semantic Communication Empowered Collaborative Perception in Constrained Networks," \textit{IEEE Wireless Commun. Lett.}, vol. 14, no. 3, pp. 701-705, Mar. 2025.

    \bibitem{ref:latent_multimodal_dynamics}
    C.B. Chaaya, A.M. Girgis, and M. Bennis, "Learning Latent Multimodal Dynamics for Optimized Resource Planning," \textit{IEEE Trans. Wirel. Commun.}, vol. 25, pp. 9591-9607, Dec. 2025.

    \bibitem{ref:jepa_msac_2026}
    C. Zheng, J. He, G. Cai, N. Li, M. Bennis, H. Wymeersch, and M. Debbah, "JEPA-MSAC: A Joint-Embedding Predictive Architecture for Multimodal Sensing-Assisted Communications," \textit{arXiv:2603.29796}. [Online]. Available: https://arxiv.org/abs/2603.29796

    \bibitem{ref:wireless_jepa_2026}
    V. Chu, O. Mashaal, and H. Abou-Zeid, "WirelessJEPA: A Multi-Antenna Foundation Model using Spatio-temporal Wireless Latent Predictions," \textit{arXiv:2601.20190}. [Online]. Available: https://arxiv.org/abs/2601.20190

    \bibitem{ref:ijepa}
    M. Assran, Q. Duval, I. Misra, P. Bojanowski, P. Vincent, M. Rabbat, Y. LeCun, and N. Ballas, "Self-Supervised Learning From Images With a Joint-Embedding Predictive Architecture," \textit{in Proc. IEEE/CVF CVPR}, pp. 15619-15629, Jun. 2023.

    \bibitem{ref:v2x_jepa_2026}
    N. Mayumu, X. Deng, A.B. Bagula, S.U.R. Khan, and P. Mukala, "V2X-JEPA: Self-Supervised Multiagent Joint Embedding Predictive Architecture for Robust Vehicle-to-Everything Perception," \textit{IEEE Internet Things J.}, vol. 13, no. 8, pp. 16609-16620, Apr. 2026.

    \bibitem{ref:aoi_survey_yates}
    R.D. Yates, Y. Sun, D.R. Brown, S.K. Kaul, E. Modiano, and S. Ulukus, "Age of Information: An Introduction and Survey," \textit{IEEE J. Sel. Areas Commun.}, vol. 39, no. 5, pp. 1183-1210, May. 2021.

    \bibitem{ref:twi_physical_ai}
    A. Mishra, J.H.I.D. Souza, and P. Popovski, "Temporal Windows of Integration for Multisensory Wireless Systems as Enablers of Physical AI," \textit{arXiv:2512.09589}. [Online]. Available: https://arxiv.org/abs/2512.09589

    \bibitem{ref:noma_survey_2017}
    Z. Ding, X. Lei, G.K. Karagiannidis, R. Schober, J. Yuan, and V.K. Bhargava, "A Survey on Non-Orthogonal Multiple Access for 5G Networks: Research Challenges and Future Trends," \textit{IEEE J. Sel. Areas Commun.}, vol. 35, no. 10, pp. 2181-2195, Oct. 2017.

    \bibitem{ref:aoii_semantic}
    A. Maatouk, M. Assaad, and A. Ephremides, "The Age of Incorrect Information: An Enabler of Semantics-Empowered Communication," \textit{IEEE Trans. Wirel. Commun.}, vol. 22, no. 4, pp. 2621-2635, Apr. 2023.

    \bibitem{ref:aoi_semantic_twc2026}
    X. Han, B. Feng, Y. Wu, X.-G. Xia, W. Zhang, and S. Sun, "Age of Semantic Information-Aware Wireless Transmission for Remote Monitoring Systems," \textit{IEEE Trans. Wirel. Commun.}, vol. 25, pp. 2939-2953, Aug. 2026.

    \bibitem{ref:aoiv_semantic_tc2025}
    M. Salimnejad, M. Kountouris, A. Ephremides, and N. Pappas, "Age of Information Versions: A Semantic View of Markov Source Monitoring," \textit{IEEE Trans. Commun.}, vol. 73, no. 12, pp. 14486-14502, Dec. 2025.

    \bibitem{ref:aosi_wcnc2024}
    L. Chen and J. Gong, "Multi-source Scheduling and Resource Allocation for Age-of-Semantic-Importance Optimization in Status Update Systems," \textit{in Proc. IEEE WCNC }, pp. 1-6, Apr. 2024.

    \bibitem{ref:correlated_wiener_scheduling_2025}
    A. Li and E. Uysal, "Optimal Sampling and Scheduling for Remote Fusion Estimation of Correlated Wiener Processes," \textit{arXiv:2510.22288}. [Online]. Available: https://arxiv.org/abs/2510.22288

    \bibitem{ref:distribution_aware_aoi_lqr_2026}
    A.Y. Etcibasi, C.E. Koksal, and E. Ekici, "When Freshness Is Not Enough: Distribution-Aware Age of Information for Networked LQR Control," \textit{arXiv:2606.04361}. [Online]. Available: https://arxiv.org/abs/2606.04361

    \bibitem{ref:correlated_sources_aoi_2022}
    V. Tripathi and E. Modiano, "Optimizing Age of Information With Correlated Sources," \textit{in Proc. ACM MobiHoc}, pp. 41-50, Oct. 2022.

    \bibitem{ref:correlated_camera_aoi_2019}
    Q. He, G. D\'an, and V. Fodor, "Joint Assignment and Scheduling for Minimizing Age of Correlated Information," \textit{IEEE/ACM Trans. Netw.}, vol. 27, no. 5, pp. 1887-1900, Oct. 2019.

    \bibitem{ref:vaoi_noma_2026}
    G. Karevvanavar, R.V. Bhat, and N. Pappas, "Version AoI Optimization under Power and General Distortion Constraints in Uplink NOMA," \textit{arXiv:2603.28631}. [Online]. Available: https://arxiv.org/abs/2603.28631

    \bibitem{ref:transformer_aoi_noma_2026}
    M. Ansarifard, M.K. Sharma, K.C. Joshi, and G. Exarchakos, "Transformer Actor-Critic for Efficient Freshness-Aware Resource Allocation," \textit{arXiv:2602.22774}. [Online]. Available: https://arxiv.org/abs/2602.22774

    \bibitem{ref:semcom_noma_2025}
    I. Ahmed, Y. Sun, J. Fu, A. Köse, L. Musavian, M. Xiao, and B. Özbek, "Semantic Communications in 6G: Coexistence, Multiple Access, and Satellite Networks," \textit{IEEE Commun. Stand. Mag.}, vol. 9, no. 4, pp. 58-64, Dec. 2025.

    \bibitem{ref:aoii_semantic_noma_xr_2023}
    J. Chen, J. Wang, C. Jiang, and J. Wang, "Age of Incorrect Information in Semantic Communications for NOMA Aided XR Applications," \textit{IEEE J. Sel. Top. Signal Process.}, vol. 17, no. 5, pp. 1093-1105, Sep. 2023.

    \bibitem{ref:sic_noma_2020}
    Z. Ding, R. Schober, and H.V. Poor, "Unveiling the Importance of SIC in NOMA Systems---Part I: State of the Art and Recent Findings," \textit{IEEE Commun. Lett.}, vol. 24, no. 11, pp. 2373-2377, Nov. 2020.

    \bibitem{ref:uplink_noma_sic}
    X. Mu and Y. Liu, "Exploiting Semantic Communication for Non-Orthogonal Multiple Access," \textit{IEEE J. Sel. Areas Commun.}, vol. 41, no. 8, pp. 2563-2576, Aug. 2023.

    \bibitem{ref:cantelli_1928}
    F.P. Cantelli, "Sui confini della probabilit\`a," \textit{in Proc. Int. Congr. Math.}, vol. 6, pp. 47-59, 1928.

    \bibitem{Karp1972}
    R. M. Karp, "Reducibility among combinatorial problems," in \textit{Raymond E. Miller and James W. Thatcher, editors, Complexity of Computer Computations,}, Plenum Press, pp. 85-103, 1972.

    \bibitem{ref:warehouse_dataset_zenodo}
    H. P. Madushanka, S. Samarakoon, and M. Bennis, "Warehouse MultiCam RF Dataset," Zenodo, [Online]. Available: https://zenodo.org/records/20315129

    \bibitem{ref:dinov2}
    M. Oquab, T. Darcet, T. Moutakanni, H. Vo, M. Szafraniec, V. Khalidov, P. Fernandez, D. Haziza, F. Massa, A. El-Nouby, M. Assran, N. Ballas, W. Galuba, R. Howes, P.Y. Huang, S.W. Li, I. Misra, M. Rabbat, V. Sharma, G. Synnaeve, H. Xu, H. Jegou, J. Mairal, P. Labatut, A. Joulin, and P. Bojanowski, "DINOv2: Learning Robust Visual Features without Supervision," \textit{arXiv:2304.07193}. [Online]. Available: https://arxiv.org/abs/2304.07193

\end{thebibliography}
\end{document}